\let\accentvec\vec
\documentclass{llncs}

\let\vec\accentvec
\usepackage{tikz,tabularx}
\usepackage{amsmath}
\usepackage{amssymb}

\usepackage{tabularx,multirow,multicol}
\usepackage{booktabs}
\usepackage{colortbl}
\usepackage{ifthen,rotating}

\usetikzlibrary{calc,decorations,shapes}
\pgfdeclarelayer{background}
\pgfdeclarelayer{foreground}
\pgfsetlayers{background,main,foreground}

\newcommand{\all}{\forall}
\newcommand{\eq}{\Leftrightarrow}

\newcommand{\impl}{\Rightarrow}
\newcommand{\eg}{{e.\,g.}}

\newcommand{\ama}{\[ \begin{aligned}}
\newcommand{\ema}{\end{aligned} \]}
\newcommand{\set}[1]{\left\{#1\right\}}

\newcommand{\setc}[2]{\left\{\left.#1\ \right|\ #2\right\}}
\newcommand{\eps}{\varepsilon}

\newcommand{\R}{\mathbb{R}}

\tikzstyle{request}=[draw, circle, color=black, fill=black, inner sep = 0.1em, minimum size = 0.55cm, text=white]
\tikzstyle{server}=[draw, circle, color=black, fill=none, inner sep = 0.1em, minimum size = 0.65cm, line width = 1mm]
\tikzstyle{matching}=[ultra thick,blue,->,>=stealth]
\newcommand{\ie}{{i.\,e.}}

\definecolor{Xmatheonlightblue}{HTML}{2B66BF}
\definecolor{Xmatheonblue}{HTML}{2659A6} 
\definecolor{Xmatheondarkblue}{HTML}{1A3B74} 
\definecolor{Xmatheonverydarkblue}{HTML}{0D1F39} 

\definecolor{darkgreen}{rgb}{0.0,0.6,0.0}
\definecolor{Apricot}     {cmyk}{0,0.32,0.52,0}
\definecolor{Aquamarine}  {cmyk}{0.82,0,0.30,0}
\definecolor{CadetBlue}   {cmyk}{0.62,0.57,0.23,0}
\definecolor{DarkGray}    {gray}{0.2}
\definecolor{DarkGreen}   {rgb}{0,0.5,0}
\definecolor{ForestGreen} {cmyk}{0.91,0,0.88,0.12}
\definecolor{Gold}        {rgb}{1.,0.84,0.}
\definecolor{Goldenrod}   {cmyk}{0,0.10,0.84,0}
\definecolor{IndianRed}   {rgb}{0.8,0.36,0.36}
\definecolor{Lavender}    {cmyk}{0,0.48,0,0}
\definecolor{LemonChiffon}{rgb}{1.,0.98,0.8}
\definecolor{LightBlue}   {rgb}{0.68,0.85,0.9}
\definecolor{LightCyan}   {rgb}{0.88,1.,1.}
\definecolor{LightGray}   {gray}{0.92}
\definecolor{LightYellow} {rgb}{1.,1.,0.88}
\definecolor{Melon}       {cmyk}{0,0.46,0.50,0}
\definecolor{NavyBlue}    {cmyk}{0.94,0.54,0,0}
\definecolor{Orange}      {rgb}{1.,0.65,0.}
\definecolor{PaleGreen}   {rgb}{0.6,0.98,0.6}
\definecolor{PaleGreenB}  {rgb}{0.9,1,0.9}
\definecolor{Peach}       {cmyk}{0,0.50,0.70,0}
\definecolor{PeachPuff}   {rgb}{1.0,0.85,0.73}
\definecolor{PineGreen}   {cmyk}{0.92,0,0.59,0.25}
\definecolor{Pink}        {rgb}{1.,0.75,0.8}
\definecolor{RoyalBlue}   {cmyk}{1,0.50,0,0}
\definecolor{SeaGreen}    {cmyk}{0.69,0,0.50,0}
\definecolor{Salmon}      {cmyk}{0,0.53,0.38,0}
\definecolor{Sepia}       {cmyk}{0,0.83,1,0.70}
\definecolor{SlateBlue}   {rgb}{0.42,0.35,0.8}
\definecolor{Thistle}     {rgb}{0.85,0.75,0.85}
\definecolor{Turquoise}   {cmyk}{0.85,0,0.20,0}
\definecolor{Violet}      {cmyk}{0.79,0.88,0,0}
\definecolor{YellowOrange}{cmyk}{0,0.42,1,0}

\tikzstyle{node}=[draw, circle, color=black, fill=white, inner sep = 0.1em, minimum size = 0.85cm]

\tikzstyle{blacknode}=[circle, fill=black, text=white, inner sep = 0.1em, minimum size = 0.6cm]

\tikzstyle{variableSource}=[draw, circle, color=black, fill=black, inner sep = 0.1em, minimum size = 0.75cm, text=white]
\tikzstyle{variableSink}=[draw, circle, color=black, fill=none, inner sep = 0.1em, minimum size = 0.75cm]
\tikzstyle{variableNode}=[draw, circle, color=black, fill=gray, inner sep = 0.1em, minimum size = 0.75cm]
\tikzstyle{shortEdge}=[ultra thick,-,>=stealth,dashed]
\tikzstyle{longEdge}=[ultra thick,-,>=stealth]
\tikzstyle{edge}=[ultra thick,-,>=stealth]

\newcommand{\splitvariableblock}[3]{
 \begin{scope}[xshift=#1cm,yshift=#2cm]
  \draw[draw=black, fill=LightGray] (-2.25,0.50) rectangle (2.25,-3.5);
  \node[variableSource,label={[label distance=-1mm]135:\tiny $-1$}] (ta#3) at (-0.5,0) {$d_{#3}^-$};
  \node[variableSource,label={[label distance=-1mm]45:\tiny $-1$}] (tb#3) at (0.5,0) {$\overline{d}_{#3}^-$};
  \node[variableSink,label={[label distance=-1mm]45:\tiny $2$}] (s#3) at (0,-1.5) {$d_{#3}^+$};
  \node[variableNode] (a#3) at (-1.75,0) {$x_{#3}^1$};
  \node[variableNode] (b#3) at (1.75,0) {$\overline{x}_{#3}^1$};
  \node[variableNode] (c#3) at (-1.75,-1.5) {$x_{#3}^2$};
  \node[variableNode] (d#3) at (1.75,-1.5) {$\overline{x}_{#3}^2$};
  \draw[shortEdge] (a#3) -- (c#3);
  \draw[shortEdge] (b#3) -- (d#3);
  \draw[longEdge] (s#3) -- (c#3);
  \draw[longEdge] (s#3) -- (d#3);
  \draw[longEdge] (a#3) -- (ta#3);
  \draw[longEdge] (b#3) -- (tb#3);
  \node[anchor=center] (text) at (0,-0.75) {\small Variable $x_{#3}$};
  \node[variableSource,label={[label distance=-1mm]45:$*$}] (tta#3) at (-1.75,-3) {$x_{#3}^-$};
  \node[variableSource,label={[label distance=-1mm]135:$*$}] (ttb#3) at (1.75,-3) {$\overline{x}_{#3}^-$};
  \draw[longEdge] (c#3) -- (tta#3);
  \draw[longEdge] (d#3) -- (ttb#3);
 \end{scope}
}

\newcommand{\extvariableblock}[3]{
 \begin{scope}[xshift=#1cm,yshift=#2cm]
  \draw[draw=black, fill=LightGray] (-3.0,0.90) rectangle (3.0,-2.75);
  \node[variableSource,label={[label distance=-1mm]90:\tiny $-\frac{1}{2}(C^2+C)$}, minimum size=0.6cm] (ta#3) at (-1.25,0) {\small $d_{#3}^-$};
  \node[variableSource,label={[label distance=-1mm]90:\tiny $-\frac{1}{2}(C^2+C)$}, minimum size=0.6cm] (tb#3) at (1.25,0) {\small $\overline{d}_{#3}^-$};
  \node[variableSink,label={[label distance=-1mm]90:\tiny $C$}, minimum size=0.6cm] (s#3) at (0,-2.0) {\small $d_{#3}^+$};
  \node[variableSink,label={[label distance=-1mm]270:\tiny $C^2$}, minimum size=0.6cm] (ss#3) at (0,0) {\small $\hat{d}_{#3}^+$};
  \node[variableNode, minimum size=0.6cm] (a#3) at (-2.5,0) {\small $x_{#3}^1$};
  \node[variableNode, minimum size=0.6cm] (b#3) at (2.5,0) {\small $\overline{x}_{#3}^1$};
  \node[variableNode, minimum size=0.6cm] (c#3) at (-2.5,-2.0) {\small  $x_{#3}^2$};
  \node[variableNode, minimum size=0.6cm] (d#3) at (2.5,-2.0) {\small  $\overline{x}_{#3}^2$};
  \draw[longEdge,thick,dashed] (a#3) -- (c#3);
  \draw[longEdge,thick,dashed] (b#3) -- (d#3);
  \draw[longEdge,thick,dashed] (s#3) -- (c#3);
  \draw[longEdge,thick,dashed] (s#3) -- (d#3);
  \draw[longEdge] (ss#3) -- (ta#3);
  \draw[longEdge] (ss#3) -- (tb#3);
  \draw[longEdge,thick,dashed] (a#3) -- (ta#3);
  \draw[longEdge,thick,dashed] (b#3) -- (tb#3);
  \node[anchor=center] (text) at (0,-1.0) {\small Variable $x_{#3}$};
 \end{scope}
}

\DeclareMathOperator{\excess}{ex}

\newcommand{\orient}[1]{\overrightarrow{#1}}
\newcommand{\orientN}{\overrightarrow{N}}

\begin{document}

\bibliographystyle{plain}

\title{Graph Orientation and Flows Over Time\thanks{Supported by the DFG Priority Program ``Algorithms for Big Data'' (SPP 1736) and by the DFG Research Center \textsc{Matheon} ``Mathematics for key technologies'' in Berlin. An extended abstract will appear in the proceedings of the 25th International Symposium on Algorithms and Computation (ISAAC '14).}
}

\titlerunning{The Price of Orientation: On the Effects of Directing a Dynamic Flow Network}

\author{Ashwin Arulselvan\inst{1}\thanks{The work was performed while the author was working at TU Berlin.} \and Martin Gro\ss\inst{2} \and Martin Skutella\inst{2}}

\institute{Department of Management Science, University of Strathclyde\\ \email{ashwin.arulselvan@strath.ac.uk}\and
Institut f\"ur Mathematik, TU Berlin, Str. des 17.~Juni 136, 10623 Berlin, Germany\\ \email{gross,skutella@math.tu-berlin.de}}

\authorrunning{A. Arulselvan, M. Gro\ss\ and M. Skutella}

\maketitle

\begin{abstract}
Flows over time are used to model many real-world logistic and routing problems. The networks underlying such problems -- streets, tracks, etc. -- are inherently undirected and directions are only imposed on them to reduce the danger of colliding vehicles and similar problems. Thus the question arises, what influence the orientation of the network has on the network flow over time problem that is being solved on the oriented network. In the literature, this is also referred to as the \emph{contraflow} or \emph{lane reversal} problem.

We introduce and analyze the \emph{price of orientation}: How much flow is lost in any orientation of the network if the time horizon remains fixed? We prove that there is always an orientation where we can still send $\frac13$ of the flow and this bound is tight. For the special case of networks with a single source or sink, this fraction is~$\frac12$ which is again tight. We present more results of similar flavor and also show non-approximability results for finding the best orientation for single and multicommodity maximum flows over time.
\end{abstract}

\section{Introduction}
\label{sec:intro}
 
Robbins~\cite{Ro:39} studied the problem of orienting streets as early as 1939, motivated by the problem of controlling congestion by making streets of a city one-way during the weekend. He showed that a strongly connected digraph could be obtained by orienting the edges of an undirected graph if and only if it is 2-edge connected. 
 
The problem of prescribing or changing the direction of road lanes is a strategy employed to mitigate congestion during an emergency situation or at rush hour. This is called a \emph{contraflow} problem (or sometimes \emph{reversible flow} or \emph{lane re\-ver\-sal} problem). Contraflows are an important tool for hurricane evacuation \cite{Wo:01}, and in that context the importance of modeling time has become prevalent in the past decade \cite{WoUrLe:02}. It is also employed to handle traffic during rush hours~\cite{MatEtAl:11}.

\emph{Flows over time} (also referred to as \emph{dynamic flows}) have been introduced by Ford and Fulkerson~\cite{FoFu:62} and extend the classic notion of static network flows. They can model a time aspect and are therefore better suited to represent real-world phenomena such as traffic, production flows or evacuations. For the latter, \emph{quickest} flows (over time) are the model of choice.  They are based on the idea that a given number of individuals should leave a dangerous area as quickly as possible~\cite{BDK93,FLTAR98}. Such an evacuation scenario is modeled by a network, with nodes representing locations. Nodes containing evacuees are denoted as \emph{sources} while the network's \emph{sinks} model safe areas or exits. For networks with multiple sources and sinks, quickest flows are also referred to as \emph{quickest transshipments}~\cite{Ho:95} and the problem of computing them is called an \emph{evacuation problem}~\cite{Tjandra2003}. A strongly polynomial algorithm for the quickest flow problem was described in~\cite{HOPTAR00}.
For a more extensive introduction to flows over time, see~\cite{-S09}. 

In this paper, we are interested in combining the orientation of a network with flows over time -- we want to orient the network such that the orientation is as beneficial as possible for the flow over time problem. We will assume that we can orient edges in the beginning, and cannot change the orientation afterwards. The assumption is reasonable in an evacuation setting as altering the orientation in the middle of an evacuation process can be difficult or even infeasible, depending on the resources available. We also assume that each edge has to be routed completely in one direction -- but this will not impose any restriction to our modeling abilities, as we can model lanes with parallel edges if we want to orient them individually.

If there is only a single source and sink, we can apply the algorithm of Ford and Fulkerson~\cite{FoFu:62} to obtain an orientation and a solution. Furthermore, it was shown that finding the best orientation for a quickest flow problem with multiple sources and sinks is NP-hard ~\cite{KiSh:05,RebEtAl:10}. Due to the hardness of the problem, heuristic and simulation tools are predominantly used in practice~\cite{KiSh:05,TuZi:04,TuZi:06,Wo:01}.

\paragraph{\textbf{Our Contribution.}}
In Section~\ref{sec:existence} we study the \emph{price of orientation} for networks with single and multiple sources and sinks, \ie, we deal with the following questions: How much flow is lost in any orientation of the network given a fixed time horizon? And how much longer do we need in any orientation to satisfy all supplies and demands, compared to the undirected network?

To our knowledge, the price of orientation has not been studied for flows over time so far. It follows from the work of Ford and Fulkerson~\cite{FoFu:62} that for $s$-$t$-flows over time the price of orientation is~$1$: Ford and Fulkerson proved that a maximum flow over time can be obtained by temporally repeating a static min-cost flow and thus uses every edge in one direction only. The latter property no longer holds if there is more than one source or sink; see Fig.~\ref{fig:qtfao}.

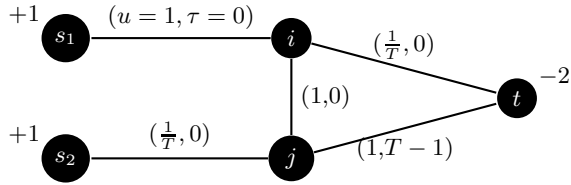
\begin{figure}[hbt]
\begin{center}
\begin{tikzpicture}[yscale=0.8,>=stealth]
	\node (s2) at (6,5)   [scale=1,circle,fill,minimum size=4mm,very thick,label={[label distance=-1mm]160: $+1$}] {\textcolor{white}{$s_2$}};	
	\node (s1) at (6,7)   [scale=1,circle,fill,minimum size=4mm,very thick,label={[label distance=-1mm]160: $+1$}] {\textcolor{white}{$s_1$}};	
	\node (t) at (12,6)   [scale=1,circle,fill,minimum size=4mm,very thick, label={[label distance=-1mm]20: $-2$}] {\textcolor{white}{$t$}};	
	\node (j) at (9,5)   [scale=1,circle,fill,minimum size=4mm] {\textcolor{white}{$j$}};	
	\node (i) at (9,7)   [scale=1,circle,fill,minimum size=4mm] {\textcolor{white}{$i$}};
	\draw [-, thick] (s1)--(i) node [midway, above]{$(u=1,\tau=0)$};
	\draw [-, thick] (s2)--(j)node [midway, above]{($\frac{1}{T},0$)};
	\draw [-, thick] (j)--(i)node [midway, right]{(1,0)};
	\draw [-, thick] (j)--(t)node [midway, below]{(1,$T-1$)};
	\draw [-, thick] (t)--(i)node [midway, above]{($\frac{1}{T},0$)};
  \end{tikzpicture}
 \end{center}

 \vspace{-4ex}
 \caption{An instance with time horizon~$T$ where flow has to use edge $\{i,j\}$ in both directions: At most $\frac{1}{T}$ units of flow can reach the sink via path $s_2,j,t$ due to the capacity of edge $\{s_2,j\}$ and the transit time of edge $\{j,t\}$. Thus, a meaningful amount of flow from source $s_2$ can only be sent via path $s_2,j,i,t$ and thus blocks edge $\{i,t\}$. As a consequence, flow originating at source $s_1$ needs to take the path $s_1,i,j,t$.}
 \label{fig:qtfao}
\end{figure}

We are able to give tight bounds for the price of orientation with regard to the flow value, and we show that the price of orientation with regard to the time horizon cannot be smaller than linear in the number of nodes. Table~\ref{overviewtable} shows an overview of our results. 
\begin{table}[bht]
 \centering 
  \begin{tabular}{c@{\hspace{5mm}}c@{\hspace{8mm}}c@{\hspace{3mm}}c@{\hspace{8mm}}c@{\hspace{3mm}}c}
  \toprule
  Sources & Sinks   & \multicolumn{2}{c}{Flow Value}                                                   & \multicolumn{2}{c}{Time}  \\ 
  \cmidrule{3-4}\cmidrule{5-6}
          &         & Price & Reference                                                                & Price & Reference         \\
  \midrule
  1       & 1       & 1   & Ford, Fulkerson~\cite{FoFu:62}                                          & 1                         & Ford, Fulkerson~\cite{FoFu:62}\\
  2+      & 1       & 2   & Theorem~\ref{corollary:upperboundsingle} & $\Omega(n)$               & Theorem~\ref{thm:price:time}\\
  1       & 2+      & 2   & Theorem~\ref{corollary:upperboundsingle} & $\Omega(n)$               & Theorem~\ref{thm:price:time}\\
  2+      & 2+      & 3   & Theorem~\ref{theorem:flowpriceupperbound}, \ref{theorem:lowerboundflow}    & $\Omega(n)$               & Theorem~\ref{thm:price:time}\\
  \bottomrule
 \end{tabular}\\[1.5mm] 
 \caption{An overview of price of orientation results.\label{overviewtable}}
\end{table}
Our main result is the tight bound of 3 on the flow price of orientation for the multiple sources and sinks case. We describe an algorithm that is capable of simulating balances through capacities of auxiliary edges. This allows us to transform a problem with supplies and demands to the much simpler case of a single source with unbounded supply and a single sink with unbounded demand. We characterize the properties that the capacities of the auxiliary edges should have for a good approximation, and describe how they can be obtained using an iterative approach that uses Brouwer fixed-points. On the negative side, we give an instance whose price of orientation is not better than~3.

Since we have two ways to pay the price of orientation -- decreasing the flow value or increasing the time horizon -- the question arises whether it might be desirable to pay the price partly as flow value and partly as time horizon. We prove that by doing so, we can achieve a bicriteria-price of $2 / 2$ for the case of multiple sources and sinks, \ie, we can send at least half the flow value in twice the amount of time. 

In Section~\ref{sec:hardness} we analyze the complexity of finding the best orientation to minimize the loss in time or flow value for a specific instance. We are able to show that these problems cannot be approximated with a factor better than 2, unless $P = NP$. Furthermore, we extend this to two multicommodity versions of this problem and show that these become inapproximable, unless $P = NP$.

\vspace{-2mm}
\section{Preliminaries}
\label{sec:prelim}

\paragraph{Networks and Orientations.}

An \emph{undirected network over time} $N$ consists of an undirected graph $G$ with a set of nodes $V(G)$, a set of edges $E(G)$, capacities $u_e \geq 0$ and transit times $\tau_e \geq 0$ on all edges $e \in E(G)$, balances $b_v$ on all nodes $v \in V(G)$, and a time horizon $T \geq 0$. For convenience, we define $V(N) := V(G), E(N) := E(G)$. The capacity $u_e$ is interpreted as the maximal \emph{inflow rate} of edge $e$ and flow entering an edge $e$ with a transit time of $\tau_e$ at time $\theta$ leaves $e$ at time $\theta+\tau_e$. We extend the edge and node attributes to sets of edges and nodes by defining: $u(E) := \sum_{e\in E} u_e$, $\tau(E) := \sum_{e\in E} \tau_e$ and $b(V) := \sum_{v\in V} b_v$. We denote the set of edges incident to a node $v$ by $\delta(v)$.

We define $S^+ := \setc{v \in V(G)}{b_v > 0}$ as the set of nodes with positive balance (also called \emph{supply}), which we will refer to as \emph{sources}. Likewise, we define $S^- := \setc{v \in V(G)}{b_v < 0}$ as the set of nodes with negative balance (called \emph{demand}), which we will refer to as \emph{sinks}. Additionally, we assume that $\sum_{v \in V(G)} b_v = 0$ and define $B := \sum_{v \in S^+} b_v$. To define a \emph{directed} network over time, replace the undirected graph with a directed one. In a directed network, we denote the set of edges leaving a node $v$ by $\delta^+(v)$ and the set of edges entering $v$ by $\delta^-(v)$ for all $v \in V(G)$.

An \emph{orientation} $\orientN$ of an undirected network over time $N$ is a \emph{directed network over time} $\orientN = (\orient{G}, \orient{u}, b, \orient{\tau}, T)$, such that $\orient{G}$, $\orient{u}$ and $\orient{\tau}$ are orientations of $G$, $u$ and $\tau$, respectively. This means that for every edge $\set{v,w} \in E(G)$ there is either $(v,w)$ or $(w,v)$ in $E(\orient{G})$ (but not both) and (assuming $(v,w) \in E(\orient{G})$) $\orient{u}_{(v,w)} = u_{\set{v,w}}$ and $\orient{\tau}_{(v,w)} = \tau_{\set{v,w}}$. Recall that we can use parallel edges if we want to model streets with multiple lanes -- each parallel edge can then be oriented individually.

\vspace{-0.5ex}
\paragraph{Flows over Time.}

A \emph{flow over time} $f$ in a directed network over time $N = (G, u, b, \tau, T)$ assigns a Lebesgue-integrable flow rate function $f_e: [0,T) \to \R^+_0$ to every edge $e\in E(G)$. We assume that no flow is left on the edges after the time horizon, \ie, $f_e(\theta) = 0$ for all $\theta \geq T - \tau_e$. The flow rate functions $f_e$ have to obey \emph{capacity constraints}, \ie, $f_e(\theta) \leq u_e$ for all $e\in E, \theta \in [0,T)$. Furthermore, they have to satisfy \emph{flow conservation constraints}. For brevity, we define the \emph{excess} of a node as the difference between the flow reaching the node and leaving it:
  $\excess_f(v,\theta) := \sum_{e\in \delta^-(v)} \int_{0}^{\theta-\tau_e} f_e(\xi)\ d\xi - \sum_{e\in \delta^+(v)} \int_{0}^{\theta} f_e(\xi)\ d\xi.$
Additionally, we define $\excess(v) := \excess(v,T)$. Then we can write the flow conservation constraints as 
\ama
 \excess(v)=0, \excess(v,\theta) &\geq 0 &\text{for all } v\in V(N) \backslash (S^+ \cup S^-), \theta \in [0,T),\\
 0 \geq \excess(v,\theta)&\geq -b_v &\text{for all } v\in S^+, \theta \in [0,T),\\
 0 \leq \excess(v,\theta)&\leq -b_v &\text{for all } v\in S^-, \theta \in [0,T).\\
\ema
The \emph{value} $|f|_\theta$ of a flow over time $f$ until time $\theta$ is the amount of flow that has reached the sinks until time $\theta$: $|f|_\theta := \sum_{s^- \in S^-} \excess_f(s^-,\theta)$ with  $\theta \in [0, T]$. For brevity, we define $|f| := |f|_T$. 

We define flows over time in undirected networks over time $N$ by transforming $N$ into a directed network $N'$, using the following construction. We replace every undirected edge $e = \set{v,w} \in E(N)$ by introducing two additional nodes $vw$, $vw'$ and edges $(v,vw)$, $(w,vw)$, $(vw,vw')$, $(vw',v)$, $(vw',w)$. We set $u_{(vw,vw')} = u_e$ and $\tau_{(vw,vw')} = \tau_e$, the rest of the new edges gets zero transit times and infinite capacities. This transformation replaces all undirected edges with directed edges, giving us the directed network $N'$. Every flow unit that could have used $\set{v,w}$ from either $v$ to $w$ or $w$ to $v$ must now use the new edge $(vw,vw')$, which has the same attributes as $\set{v,w}$. The other four edges just ensure that $(vw,vw')$ can be used by flow from $v$ to $w$ or $w$ to $v$. Thus, whenever we consider flows over time in $N$, we interpret them as flows over time in $N'$ instead.

\vspace{-0.5ex}
\paragraph{Maximum Flows over Time.}
The \emph{maximum flow over time problem} consists of a directed or undirected network over time $N = (G, u, b, \tau, T)$ where the objective is to find a flow over time of maximum value. The sources and sinks have usually unbounded supplies and demands in this setting but it can also be studied with finite supplies and demands. In the latter case, the problem is sometimes referred to as \emph{transshipment over time} problem.

Ford and Fulkerson~\cite{FoFu:62} showed that the case of unbounded supplies and demands can be solved by a reduction to a static minimum cost circulation problem. This yields a \emph{temporally repeated} flow as an optimal solution. Such a flow $x$ is given by a family of paths $\mathcal{P}$ along which flow is sent at constant rates $x_P$, $P \in \mathcal{P}$ during the time intervals $[0,T - \tau_P)$, with $\tau_P := \sum_{e\in P} \tau_e$. The algorithm of Ford and Fulkerson obtains these paths by decomposing the solution to the minimum cost circulation problem. This algorithm has the nice property that edges are only used in one direction, as it is based on a static flow decomposition.

The \emph{maximum contraflow over time problem} is given by an undirected network over time $N = (G, u, b, \tau, T)$ and the objective is to find an orientation $\orient{N}$ of $N$ such that the value of a maximum flow over time in $\orient{N}$ is maximal over all possible orientations of $N$.

\vspace{-0.5ex}
\paragraph{Quickest Flows.}
The \emph{quickest flow problem} or \emph{quickest transshipment problem} is given by a directed or undirected network over time $N = (G, u, b, \tau)$ and the objective is to find the smallest time horizon $T$ such that all supplies and demands can be fulfilled, \ie, a flow over time with value $B$ can be sent. Hoppe and Tardos~\cite{HOPTAR00} gave a polynomial algorithm to solve this problem. However, an optimal solution to this problem might have to use an edge in both directions; see Fig.~\ref{fig:qtfao}. 

The \emph{quickest contraflow problem} is given by an undirected network over time $N = (G, u, b, \tau)$ and the objective is to find an orientation $\orient{N}$ of $N$ such that the time horizon of a quickest flow in $\orient{N}$ is minimal over all possible orientations of $N$.

\section{The Price of Orientation}
\label{sec:existence}

We study two different models for the price of orientation. The \emph{flow price of orientation} for an undirected network over time $N = (G, u, b, \tau, T)$ is the ratio between the value of a maximum flow over time $f_N$ in $N$ and maximum of the values of maximum flows over time $f_{\orient{N}}$ in orientations $\orient{N}$ of $N$:
$$ |f_N| / \max_{\orient{N} \text{ orientation } of N} |f_{\orient{N}}|.$$

Similarly, the \emph{time price of orientation} for an undirected network over time $N = (G, u, b, \tau)$ is the ratio between the minimal time horizon $T(f_{\orient{N}})$ of a quickest flow $f_{\orient{N}}$ in an orientation $\orient{N}$ of $N$ and the time horizon $T(f_N)$ of a quickest flow over time $f_N$ in $N$:
$$ \min_{\orient{N} \text{ orientation } of N} T(f_{\orient{N}}) / T(f_N).$$

\subsection{Price in Terms of Flow Value}

In this subsection, we will examine the flow price of orientation. We will see that orientation can cost us two thirds of the flow value in some instances, but not more.

\begin{theorem}
 \label{theorem:flowpriceupperbound}
 Let $N = (G, u, b, \tau, T)$ be an undirected network over time, in which $B$ units of flow can be sent within the time horizon $T$. Then there exists an orientation $\orient{N}$ of $N$ in which at least $B/3$ units of flow can be sent within time horizon $T$.
\end{theorem}

\begin{proof}
 The idea of this proof is to simplify the instance, such that a temporally repeated solution can be found. Such a solution gives us an orientation that we can use, if the simplification does not cost us too much in terms of flow value. We will achieve this by simulating the balances using additional edges and capacities, creating a maximum flow over time problem which permits a temporally repeated solution. Then we show that the resulting maximum flow over time problem is close enough to the original problem for our claim to follow. 
 
\vspace{-1ex}
\paragraph{Simulating the balances.} We achieve this by adding a super source $s$ and a super sink $t$ to the network, resulting in an undirected network over time $N' = (G', u', \tau', s, t, T)$ with $V(G') := V(G) \cup \set{s,t}$, $E(G') := E(G) \cup \{\set{s,s^+}\}|\{s^+ \in S^+\} \cup \setc{\set{s^-,t}}{s^- \in S^-}$, $u'_e := u_e$ for $e \in E(G)$ and $\infty$ otherwise, $\tau'_e := \tau_e$ for $e \in E(G)$ and $0$ otherwise.
 We refer to the newly introduced edges of $E(G') \setminus E(G)$ as \emph{auxiliary} edges. Furthermore, we sometimes refer to an auxiliary edge by the unique terminal node it is adjacent to and write $u_v$ for $u_e, e = (s,v)$, $f_v$ for $f_e, e = (s,v)$ and so on. An illustration of this construction can be found in Fig.~\ref{fig:superterminals}.
\begin{figure}[tb]
  \centering
  \begin{tikzpicture}
   \node[fill,circle,minimum size=6mm] (s) at (0,0) {\textcolor{white}{$s$}};
   \node[draw,circle,minimum size=6mm] (a) at (2,-1) {}
    edge[thick,-,dashed] (s);
   \node[draw,circle,minimum size=6mm] (b) at (2,-0) {}
    edge[thick,-,dashed] (s);
   \node[draw,circle,minimum size=6mm] (c) at (2,1) {}
    edge[thick,-] (b) 
    edge[thick,-,dashed] (s);
   \node[draw,circle,minimum size=6mm] (d) at (3.5,0.8) {}
    edge[thick,-] (b)
    edge[thick,-] (c);
   \node[draw,circle,minimum size=6mm] (e) at (4.0,-1.1) {}
    edge[thick,-] (a)
    edge[thick,-] (b)
    edge[thick,-] (d);
   \node[draw,circle,minimum size=6mm] (f) at (4.7,1.2) {}
    edge[thick,-] (d);
   \node[draw,circle,minimum size=6mm] (g) at (5.3,0.1) {}
    edge[thick,-] (d)
    edge[thick,-] (e)
    edge[thick,-] (f);
   \node[draw,circle,minimum size=6mm] (h) at (7.0,0.5) {}
    edge[thick,-] (f)
    edge[thick,-] (g);
   \node[draw,circle,minimum size=6mm] (i) at (7.0,-0.5) {}
    edge[thick,-] (e)
    edge[thick,-] (g);
   \node[fill,circle,minimum size=6mm] (t1) at (9,0) {\textcolor{white}{$t$}}
    edge[thick,-,dashed] (h)
    edge[thick,-,dashed] (i);
  \end{tikzpicture}
  \caption{The modified network consisting of the original network (white), the superterminals (black) and the dashed auxiliary edges. \label{fig:superterminals}}
\end{figure}
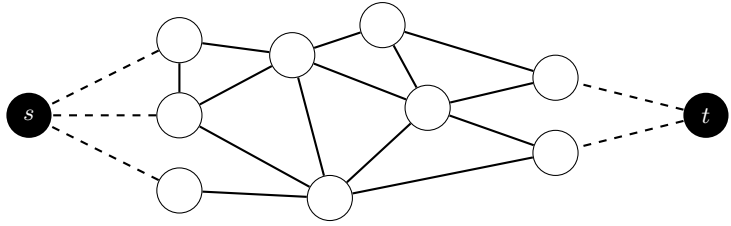

 The network $N'$ describes a maximum flow over time problem which has an optimal solution that is a temporally repeated flow, which uses each edge only in one direction during the whole time interval $[0,T)$. Thus, there is an orientation $\orient{N'}$ such that the value of a maximum flow over time in $N'$ is the same as in $\orient{N'}$. However, an optimal solution for $N'$ will generally be infeasible for $N$, since there are no balances in $N'$.

 Thus, we need to modify $N'$ such that balances of $N$ are respected -- but without using actual balances. This leaves us the option to modify the capacities of the auxiliary edges. In the next step, we will show that we can always find capacities that enforce that the balances constraints are satisfied and have nice properties for bounding the loss in flow value incurred by the capacity modification. These properties are then used in the last step to complete the proof.

\vspace{-1ex}
\paragraph{Enforcing balances by capacities for auxiliary edges.} In this step, we show that we can choose capacities for the auxiliary edges in such a way that there is a maximum flow over time in the resulting network that respects the original balances. Choosing finite capacities for some of the auxiliary edges will -- in general -- reduce the maximum flow value that can be sent, though. In order to bound this loss of flow later on, we need capacities with nice properties, that can always be found. 

 \begin{lemma}
  \label{lemma:capacities}
  There are capacities $u''_e$ that differ from $u'_e$ only for the auxiliary edges, such that the network $N'' = (G', u'', \tau', s, t, T)$ has a temporally repeated maximum flow over time $f$ with the following properties
  \begin{itemize}
   \item the balances of the nodes in the original setting are respected:\\
     $ |f_v| := \int_0^T f_v(\theta)\ d\theta  \leq |b_v| \quad \all v \in S^+ \cup S^-, $
   \item and that terminals without tightly fulfilled balances have auxiliary edges with unbounded capacity:
     $ |f_v| < |b_v| \impl u_v = \infty \quad \all v \in S^+ \cup S^-. $
  \end{itemize}
 \end{lemma}
  
 \begin{proof}
  The idea of this proof is to start with unbounded capacities and iteratively modify the capacities based on the balance and amount of flow currently going through an node, until we have capacities satisfying our needs. In order to show that such capacities exist, we apply Brouwer's fixed-point theorem on the modification function to show the existence of a fix point. By construction of the modification function, this implies the existence of the capacities.  
  
\vspace{-1ex}
\paragraph{Prerequisites for using Brouwer's fixed point-theorem.}
  We begin by defining $U := \sum_{v\in S^+} \sum_{a = (v,\cdot) \in E(G)} u_a$ as an upper bound for the capacity of auxiliary edges and we will treat $U$ and $\infty$ interchangeably from now on. This allows us to consider capacities in the interval $[0,U]$, which is convex and compact, instead of $[0,\infty)$. This will be necessary for applying Brouwer's fixed point theorem later on.

  Now assume that we have some capacities $u \in [0,U]^{S^+ \cup S^-}$ for the auxiliary edges. Since we leave the capacities for all other edges unchanged, we identify the capacities for the auxiliary edges with the capacities for all edges. Compute a maximum flow over time $f(u)$ for $(G', u, \tau', s, t, T)$ by using Ford and Fulkersons' reduction to a static minimum cost flow. For this proof, we need to ensure that small changes in $u$ result in small changes in $f(u)$, \ie, we need continuity. Thus, we will now specify that we compute the minimum cost flow by using successive computations of shortest $s$-$t$-paths. In case there are multiple shortest paths in an iteration, we consider the shortest path graph, and choose a path in this graph by using a depth-first-search that uses the order of edges in the adjacency list of the graph as a tie-breaker. The path decomposition of the minimum cost flow deletes paths in the same way. This guarantees us that we choose paths consistently, leading to the continuity that we need.

\vspace{-1ex}
\paragraph{Defining the modification function.}
  In order to obtain capacities for a maximum flow over time that respects the balances, we define a function $h: [0,U]^{S^+ \cup S^-} \to [0,U]^{S^+ \cup S^-}$ which will reduce the capacities of the auxiliary edges, if balances are not respected:
  \ama
   (h(u))_v := \min \set{U, \frac{b_v}{|f_v(u)|} u_v} \quad \all v \in S^+ \cup S^-.
  \ema
  $|f_v(u)|$ refers to the amount of flow going through the auxiliary edge of terminal $v \in S^+ \cup S^-$ in this definition. If $|f_v(u)| = 0$, we assume that the minimum is $U$. Due to our rigid specification in the maximum flow computation, $|f_v(u)|$ is continuous, and therefore $h$ is continuous as well.

\vspace{-1ex}
\paragraph{Using Brouwer's fixed-point theorem.}
  Thus, $h$ is continuous over a convex, compact subset of $\R^{S^+ \cup S^-}$. By Brouwer's fixed-point theorem it has a fixed point $\overline{u}$ with $h(\overline{u}) = \overline{u}$, meaning that for every $v \in S^+ \cup S^-$ either $u_v = U$ or $u_v = \frac{b_v}{|f_v(u)|} u_v \eq b_v = |f_v(u)|$ holds, which is exactly what we require of our capacities. \qed
 \end{proof}

 We can now choose capacities $u''$ in accordance to Lemma~\ref{lemma:capacities}, and thereby gain a maximum flow over time problem instance $N'' = (G', u'', \tau', s, t, T)$, that has a temporally repeated optimal solution which does not violate the original balances. What is left to do is to analyze by how much the values of optimal solutions for $N$ and $N''$ are apart.

\vspace{-1ex}
\paragraph{Bounding the difference in flow value between $N$ and $N''$.} We now want to show that we can send at least $B/3$ flow units in the network $N''$ with the auxiliary capacities of the previous step. For the purpose of this analysis, we partition the sources and sinks as follows. 
 \ama
  S^+_1 &:= \setc{s^+ \in S^+}{u_{s^+} < \infty}, S^+_2 := \setc{s^+ \in S^+}{u_{s^+} = \infty},\\
  S^-_1 &:= \setc{s^- \in S^-}{u_{s^-} < \infty}, S^-_2 := \setc{s^- \in S^-}{u_{s^-} = \infty}.\\
 \ema
The partitioning is also shown in Fig.~\ref{fig:partitioning}.

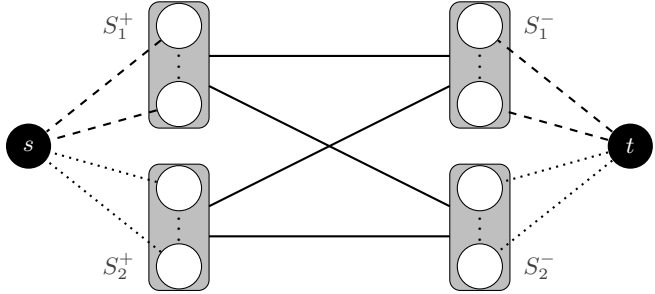
\begin{figure}[htb]
  \centering
  \begin{tikzpicture}[yscale=0.8]
   \node[fill,circle,minimum size=6mm] (s) at (0,0) {\textcolor{white}{$s$}};
   \draw[draw=black,rounded corners,fill=lightgray,name=M] (1.6,-2.4) rectangle (2.4,-0.3);
   \draw[draw=black,rounded corners,fill=lightgray,name=N] (1.6, 0.3) rectangle (2.4, 2.4);   
   \node[draw,fill=white,circle,minimum size=6mm] (a) at (2,-2) {}
    edge[thick,-,dotted] (s);
   \node[] (x) at (2,-1.25) {$\vdots$};
   \node[draw,fill=white,circle,minimum size=6mm] (b) at (2,-0.7) {}
    edge[thick,-,dotted] (s);
   \node[draw,fill=white,circle,minimum size=6mm] (c) at (2,0.7) {}
    edge[thick,-,dashed] (s);
   \node[] (x) at (2,1.45) {$\vdots$};
   \node[draw,fill=white,circle,minimum size=6mm] (d) at (2,2) {}
    edge[thick,-,dashed] (s);
   \draw[draw=black,rounded corners,fill=lightgray,name=O] (5.6,-2.4) rectangle (6.4,-0.3);
   \draw[draw=black,rounded corners,fill=lightgray,name=P] (5.6, 0.3) rectangle (6.4, 2.4);
   \draw[thick,-] (2.4, -1.5) -- (5.6, -1.5);  
   \draw[thick,-] (2.4,-1) -- (5.6, 1);
   \draw[thick,-] (2.4, 1) -- (5.6,-1);
   \draw[thick,-] (2.4, 1.5) -- (5.6, 1.5);
   \node[draw,fill=white,circle,minimum size=6mm] (e) at (6,-2) {};
   \node[] (x) at (6,-1.25) {$\vdots$};
   \node[draw,fill=white,circle,minimum size=6mm] (f) at (6,-0.7) {};
   \node[draw,fill=white,circle,minimum size=6mm] (g) at (6,0.7) {};
   \node[] (x) at (6,1.45) {$\vdots$};
   \node[draw,fill=white,circle,minimum size=6mm] (h) at (6,2) {};
   \node[fill,circle,minimum size=6mm] (t1) at (8,0) {\textcolor{white}{$t$}}
    edge[thick,-,dotted] (e)
    edge[thick,-,dotted] (f)
    edge[thick,-,dashed] (g)
    edge[thick,-,dashed] (h) ;
   \node[darkgray] (x) at (1.2,2.0) {$S^+_1$};
   \node[darkgray] (x) at (1.2,-2.0) {$S^+_2$};
   \node[darkgray] (x) at (6.8,2.0) {$S^-_1$};
   \node[darkgray] (x) at (6.8,-2.0) {$S^-_2$};
  \end{tikzpicture}
  \caption{The partitioning based on the capacities of the auxiliary edges. Dashed edges have finite capacity, dotted edges have infinite capacities. \label{fig:partitioning}}
\end{figure}

 Now let $f$ be a temporally repeated maximum flow in $N''$ that does not violate balances. Notice that the auxiliary edges to terminals in $S^+_2$ and $S^-_2$, respectively, have infinite capacity and that the supply / demand of nodes in $S^+_1$ and $S^-_1$ is fully utilized. Thus, $|f| \geq \max\set{b(S^+_1),b(S^-_1)}$. Should $b(S^+_1) \geq B/3$ or $b(S^-_1) \geq B/3$ hold, we would be done -- so let us assume that $b(S^+_1) < B/3$ and $b(S^-_1) < B/3$. It follows that $b(S^+_2) \geq 2/3 B$ and $b(S^-_2) \geq 2/3 B$ must hold in this case. 
 Now consider the network $N'$ with the terminals of $S^+_1$ and $S^-_1$ removed, leaving only the terminals of $S^+_2$ and $S^-_2$. We call this network $N'(S^+_2,S^-_2)$. Let $|f'|$ be the value of a maximum flow over time in $N'(S^+_2,S^-_2)$. Since $B$ units of flow can be sent in $N$ (and therefore $N'$ as well), we must be able to send at least $B/3$ units in $N'(S^+_2,S^-_2)$. This is due to the fact that $b(S^+_2) \geq 2/3 B, b(S^-_2) \geq 2/3 B$ -- even if $B/3$ of these supplies and demands were going to $S^-_1$ and coming from $S^+_1$, respectively, this leaves at least $B/3$ units that must be send from $S^+_2$ to $S^-_2$. Thus, $B/3 \leq |f'|$. Since the capacities of the auxiliary edges of $S^+_2$ and $S^-_2$ are infinite, we can send these $B/3$ flow units in $N''$ as well, proving this part of the claim.

 Thus, we have shown that a transshipment over time problem can be transformed into a maximum flow over time problem with auxiliary edges and capacities. If these edges and capacities fulfill the requirements of Lemma~\ref{lemma:capacities}, we can transfer solutions for the maximum flow problem to the transshipment problem such that at least one third of the total supplies of the transshipment problem can be send in the flow problem. Finally, the proof of Lemma~\ref{lemma:capacities} shows that such capacities do always exist, completing the proof.\qed
\end{proof}

Notice that the algorithm described in the proof is not necessarily efficient -- it relies on Brouwer's fixed-point theorem, and finding an (approximate) Brouwer fixed-point is known to be PPAD-complete \cite{Papadimitriou1994} and exponential lower bounds for the common classes of algorithms for this problem are known \cite{Hirsch1989}. Since the algorithm is efficient aside from finding a Brouwer fixed-point, our problem is at least not harder than finding a Brouwer fixed-point. Thus, our problem is probably not FNP-complete (with FNP being the functional analog of NP) as PPAD-completeness indicates that a problem is not FNP-complete \cite{Papadimitriou1994}. However, it is possible that the fixed-point can efficiently be found for the specific function we are interested in. One problem for finding such an algorithm is however, that changing the capacity of one auxiliary edge does not only modify the amount flow through its associated terminal but through other terminals as well -- and this change in flow value can be an increase or decrease, making monotonicity arguments problematic. 

Another potential approach could be to find a modification function for which (approximate) Brouwer fixed-points can be found efficiently. Using approximate Brouwer fixed-points would result in a weaker version of Lemma 1, where an additional error is introduced due to the approximation. This error can be made arbitrarily small by approximating the Brouwer fixed-point more closely, or by using alternative modification functions. However, finding a modification function for which an approximation of sufficient quality can be found efficiently remains an open question.

Now that we have an upper bound for the flow price of orientation and it turns out that this bound is tight. 

\begin{theorem}
 \label{theorem:lowerboundflow}
 For any $\eps > 0$, there are undirected networks over time $N = (G, u, b, \tau, T)$ in which $B$ units of flow can be sent, but at most $B/3+\eps$ units of flow can be sent in any orientation $\orient{N}$ of $N$.
\end{theorem}

\begin{proof}
In order to show this, we consider the network in Fig.~\ref{fig:lowerboundexample} with three sources and sinks where each source has to send flow to a specific sink (due to capacities and transit times) but the network topology prevents flow from more than one source-sink pair being able to be send in any orientation. 
 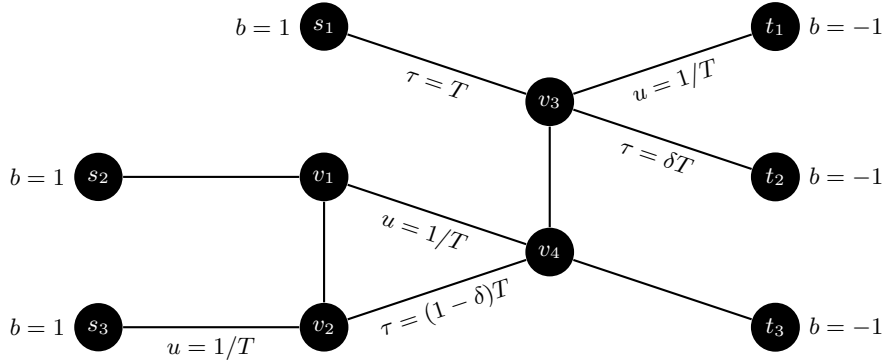
\begin{figure}[tb]
  \centering
  \begin{tikzpicture}
   \node[fill,circle,minimum size=4mm,label={west:$b = 1$}] (a1) at (0,0) {\textcolor{white}{$s_3$}};
   \node[fill,circle,minimum size=4mm,label={west:$b = 1$}] (a2) at (0,2) {\textcolor{white}{$s_2$}};
   \node[fill,circle,minimum size=4mm] (b1) at (3,0) {\textcolor{white}{$v_2$}}
    edge[thick,-] node[auto] {$u = 1/T$} (a1);
   \node[fill,circle,minimum size=4mm] (b2) at (3,2) {\textcolor{white}{$v_1$}}
    edge[thick,-] (a2)
    edge[thick,-] (b1);
   \node[fill,circle,minimum size=4mm,label={west:$b = 1$}] (b3) at (3,4) {\textcolor{white}{$s_1$}};
   \node[fill,circle,minimum size=4mm] (c1) at (6,1) {\textcolor{white}{$v_4$}}
    edge[thick,-] node[auto,sloped,pos=0.9] {$\tau = (1-\delta)T$} (b1)
    edge[thick,-] node[auto,sloped,pos=0.25] {$u = 1/T$} (b2);
   \node[fill,circle,minimum size=4mm] (c2) at (6,3) {\textcolor{white}{$v_3$}}
    edge[thick,-] node[auto,sloped,pos=0.25] {$\tau = T$} (b3)
    edge[thick,-] (c1);
   \node[fill,circle,minimum size=4mm,label={east:$b = -1$}] (d1) at (9,0) {\textcolor{white}{$t_3$}}
    edge[thick,-] (c1);
   \node[fill,circle,minimum size=4mm,label={east:$b = -1$}] (d2) at (9,2) {\textcolor{white}{$t_2$}}
    edge[thick,-] node[auto,sloped,pos=0.25] {$\tau = \delta T$}  (c2);
   \node[fill,circle,minimum size=4mm,label={east:$b = -1$}] (d3) at (9,4) {\textcolor{white}{$t_1$}}
    edge[thick,-] node[auto,sloped,pos=0.75] {$u = 1/T$} (c2);
  \end{tikzpicture}
  \caption{An undirected network where every orientation can send at most one third of the flow possible in the undirected setting. Not specified transit times and balances are 0 and not specified capacities are infinite.
 \label{fig:lowerboundexample}}
 \end{figure}

Consider the undirected network over time $N = (G, u, b, \tau, T + \eps)$ depicted in Fig.~\ref{fig:lowerboundexample}, for some $\eps > 0, \delta \in (0,1)$. 
For $\eps \leq \delta T$, we cannot send flow from $s_1$ to $t_2$ within the time horizon, and we can only send $\eps / T$ flow from $s_1$ to $t_1$. Thus, we have to orient $\set{v_3,v_4}$ as $(v_3,v_4)$ or lose the supply of $s_1$ in the case of $\eps \to 0$. Orienting $\set{v_3,v_4}$ as $(v_3,v_4)$ causes us to lose the demands of $t_1$ and $t_2$, though, resulting in only one third of the flow being able to be sent.

 Therefore let us now orient $\set{v_3,v_4}$ as $(v_4,v_3)$. Supply from $s_3$ needs to go through $\set{s_3,v_2}$ at a rate of at most $1/T$. Thus, if we were to route flow through $\set{s_3,v_2}$ and $\set{v_2,v_4}$ we can send at most $(T+\eps - (1-\delta)T) / T = \eps/T+\delta$ to $v_4$ (and the sinks) within the time horizon. For $\delta, \eps \to 0$ this converges to 0 as well. Since we already lost the supply of $s_1$, we need the supply of $s_3$ if we want to send significantly more than one unit of flow. Therefore, we would have to orient $\set{v_1,v_2}$ as $(v_1,v_2)$ to accomplish this. However, due to the capacity of $1/T$ on $\set{v_1,v_4}$ we can send at most $1+\delta/T$ flow through this edge, and one unit of this flow comes from $s_3$, leaving only $\delta/T$ units for flow from $s_2$. Thus, for $\delta,\eps \to 0$ the flow we can send converges to one. 

 In the undirected network, we can send all supplies. The supply from $s_1$ is sent to $t_3$, using $\set{v_3,v_4}$ at time $T$. The supply from $s_2$ is sent to $t_2$, via $\set{v_1,v_2}$ at time 0, $\set{v_2,v_4}$ and $\set{v_4,v_3}$ at time $(1-\delta)T$. The supply from $s_3$ is sent to $t_1$ by $\set{v_2,v_1}$, $\set{v_1,v_4}$ and $\set{v_4,v_3}$ during the time interval $(0,T)$. This completes the proof. \qed
\end{proof}

With these theorems, we have a tight bound for the flow price of orientation in networks with arbitrarily many sources and sinks. In the case of a single source and sink, we have a maximum flow over time problem and we can always find an orientation in which we can send as much flow as in the undirected network. This leaves the question about networks with either a single source or a single sink open. However, if we use the knowledge that only one source (or sink) exists in the analysis done in the proofs of Theorem~\ref{theorem:flowpriceupperbound} and Theorem~\ref{theorem:lowerboundflow}, we achieve a tight factor of 2 in these cases.

\begin{theorem}
 \label{corollary:upperboundsingle}
 Let $N = (G, u, b, \tau, T)$ be an undirected network over time with a single source or sink, in which $B$ units of flow can be sent within the time horizon $T$. Then there exists an orientation $\orient{N}$ of $N$ in which at least $B/2$ units of flow can be sent within time horizon $T$, and there are undirected networks over time for which this bound is tight.
\end{theorem}

\begin{proof}
 For this proof, we can use most of the argumentation of the proof of Theorem~\ref{theorem:flowpriceupperbound}. The differences start only in the last part, where the differences in flow value between the original network $N$ and the network with capacitated auxiliary edges $N''$ is considered. In the proof of Theorem~\ref{theorem:flowpriceupperbound}, we partitioned the sources and sinks, but now we have either a single source or a single sink which does not need to be partitioned. Let us assume now that we have a single sink, the case with a single source follows analogously. 
 We partition the sources as follows. 
 \ama
  S^+_1 &:= \setc{s^+ \in S^+}{u_{s^+} < \infty}, S^+_2 := \setc{s^+ \in S^+}{u_{s^+} = \infty},\\
 \ema
 We assume $b(S^+_1) < B/2$, because otherwise there is nothing to show. This implies that $b(S^+_2) \geq B/2$, however. We can now consider the network $N'$ with the sources $S^+_1$ removed and refer to the resulting network as $N'(S^+_2,S^-)$. Since $B$ units of flow can be sent in $N$ (and therefore $N'$ as well), we must be able to send at least $b(S^+_2)$ units in $N'(S^+_2,S^-)$, since we still have all sinks available. Because of $b(S^+_2) \geq B/2$, this proves the first part of the claim. For the second part, the lower bound, consider the construction from Theorem~\ref{theorem:lowerboundflow}.
 
  If we restrict the network described there to $s_2,s_3,v_1,v_2,v_4$ and set the balance of $v_4$ to $-2$, we can apply the same argumentation as in Theorem~\ref{theorem:lowerboundflow} to get a proof for the case of a single sink. For the case of a single source, we do something similar, but have to change something more. The result can be seen in Fig.~\ref{fig:lowerboundexample2}; the argumentation is analogous to Theorem~\ref{theorem:lowerboundflow}. \qed
 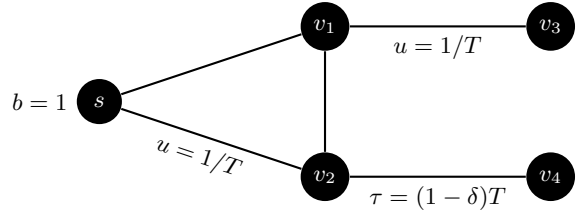
\begin{figure}[tb]
  \centering
  \begin{tikzpicture}
   \node[fill,circle,minimum size=6mm,label={west:$b = 1$}] (a) at (0,1) {\textcolor{white}{$s$}};
   \node[fill,circle,minimum size=6mm] (b1) at (3,0) {\textcolor{white}{$v_2$}}
    edge[thick,-] node[auto,sloped,pos=0.25] {$u = 1/T$} (a);
   \node[fill,circle,minimum size=6mm] (b2) at (3,2) {\textcolor{white}{$v_1$}}
    edge[thick,-] (a)
    edge[thick,-] (b1);
   \node[fill,circle,minimum size=6mm] (c1) at (6,0) {\textcolor{white}{$v_4$}}
    edge[thick,-] node[auto,sloped,pos=0.5] {$\tau = (1-\delta)T$} (b1);
   \node[fill,circle,minimum size=6mm] (c2) at (6,2) {\textcolor{white}{$v_3$}}
    edge[thick,-] node[auto,sloped,pos=0.5] {$u = 1/T$} (b2);
  \end{tikzpicture}
  \caption{An undirected network with a single source where every orientation can send at most one half of the flow possible in the undirected setting. Not specified transit times and balances are 0 and not specified capacities are infinite.
\label{fig:lowerboundexample2}}
 \end{figure} 
\end{proof}

\subsection{Price in Terms of the Time Horizon}

In this part, we examine by how much we need to extend the time horizon in order to send as much flow in an orientation as in the undirected network. It turns out that there are instances for which we have to increase the time horizon by a factor that is linear in the number of nodes. This is due to the fact that we have to send everything, which can force us to send some flow along very long detours -- this is similar to what occurs in~\cite{GrKaScSc:12}. For this reason it is not a good idea to pay the price of orientation in time alone. 

\begin{theorem}
 \label{thm:price:time}
 There are undirected networks over time $N = (G, u, b, \tau, T+1)$ with either a single source or a single sink in which $B$ units of flow can be sent within a time horizon of $T$, but it takes a time horizon of at least $(n-1)/4  \cdot T$ to send $B$ units of flow in any orientation $\orient{N}$ of $N$. This bound also holds if $G$ is a tree with multiple sources and sinks. 
\end{theorem}

\begin{proof}
 We define a family of undirected networks over time $N_k$ by
 \ama
  V(N_k) &:= \set{s_0,v_0,v_k,t_k,t} \cup \setc{s_i,t_i,v_i,w_i}{i=1,\dots,k-1}, \\
  E(N_k) &:= \set{\set{s_0,v_0},\set{v_{k-1},v_k},\set{v_k,t_k},\set{t_k,t}} \\
         &\ \ \ \cup \setc{\set{s_i,w_i},\set{t_i,w_i},\set{w_i,v_i},\set{v_i,v_{i-1}},\set{t_i,t}}{i=1,\dots,k-1}. \\ 
 \ema
 We define capacities, transit times and balances for this network by
 \ama
  u_e &:= \begin{cases} (nT)^{i-1} & e = \set{w_i,t_i}\\ \infty & \text{else} \end{cases}, \tau_e := \begin{cases} T & e = \set{v_i,v_{i-1}}\\ 0 & \text{else} \end{cases}, \\ b_v &:= \begin{cases} (nT)^i & v = s_i \\ -\sum_{i=0}^{k-1} (nT)^i & v = t\\ 0 & \text{else} \end{cases}.
 \ema
 Fig.~\ref{fig:timeboundsources} depicts such a network $N_k$.
 \begin{figure}[tb]
  \centering
  \begin{tikzpicture}
   \node[fill,circle,minimum size=7mm,label={south:$b = (nT)^0$}] (s0) at (0,0) {\textcolor{white}{$s_0$}};
   \node[fill,circle,minimum size=7mm] (v0) at (0,2.5) {\textcolor{white}{$v_0$}}
    edge[thick,-] (s0);
   \node[fill,circle,minimum size=7mm] (v1) at (3,2.5) {\textcolor{white}{$v_1$}}
    edge[thick,-] node[auto,swap] {$\tau = T$} (v0);
   \node[fill,circle,minimum size=7mm] (w1) at (3,1.25) {\textcolor{white}{$w_1$}}
    edge[thick,-] (v1);
   \node[fill,circle,minimum size=7mm,label={south:$b = (nT)^1$}] (s1) at (2.5,0) {\textcolor{white}{$s_1$}}
    edge[thick,-] (w1);
   \node[fill,circle,minimum size=7mm] (t1) at (3.5,0) {\textcolor{white}{$t_1$}}
    edge[thick,-] node[auto,swap] {$u=(nT)^0$} (w1);

   \node[fill,circle,minimum size=7mm] (v2) at (6,2.5) {\textcolor{white}{$v_2$}}
    edge[thick,-] node[auto,swap] {$\tau = T$} (v1);
   \node[fill,circle,minimum size=7mm] (w2) at (6,1.25) {\textcolor{white}{$w_2$}}
    edge[thick,-] (v2);
   \node[fill,circle,minimum size=7mm,label={south:$b = (nT)^2$}] (s2) at (5.5,0) {\textcolor{white}{$s_2$}}
    edge[thick,-] (w2);
   \node[fill,circle,minimum size=7mm] (t2) at (6.5,0) {\textcolor{white}{$t_2$}}
    edge[thick,-] node[auto,swap] {$u=(nT)^1$} (w2);    
   
   \node (v3) at (8,2.5) {}
    edge[thick,-,dashed] node[auto,swap] {$\tau = T$} (v2);
   
   \node (w3) at (8.75,1.25) {$\dots$};
   
   \node (v4) at (9,2.5) {};
   
   \node[fill,circle,minimum size=7mm] (vk) at (11,2.5) {\textcolor{white}{$v_k$}}
    edge[thick,-,dashed] node[auto,swap] {$\tau = T$} (v4);
   \node[fill,circle,minimum size=7mm] (tk) at (11,0) {\textcolor{white}{$t_k$}}
    edge[thick,-] (vk);      
    
   \node[fill,circle,minimum size=7mm,label={south:$b=-\sum_{i=0}^{k-1} (nT)^i$}] (t) at (6.25,-1.75) {\textcolor{white}{$t$}};
   \draw[thick,-] (t1) |- (t);
   \draw[thick,-] (t2) -- ++(0,-1.0) -- (t);
   \draw[thick,-] (tk) |- (t);

  \end{tikzpicture}
  \caption{An undirected network with a single sink where every orientation requires a time horizon that is larger by a factor of at least $ (n-1)/4 $ compared to the undirected setting. Not specified transit times and balances are 0 and not specified capacities are infinite.
\label{fig:timeboundsources}}
 \end{figure}
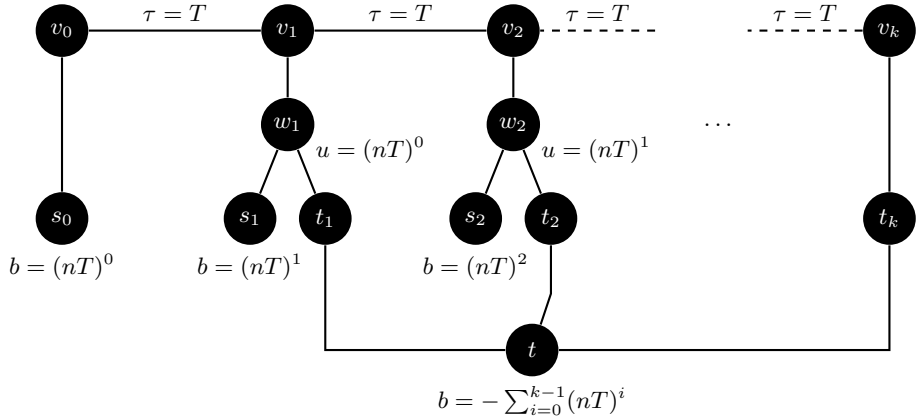
 It is possible to fulfill all supplies and demands in time $T+1$ in the undirected network, if we route the supply of source $s_i$ through $v_i$, $v_{i+1}$, $w_{i+1}$ and $t_{i+1}$ to $t$. However, this requires using the $\set{v_i,w_i}$-edges in both directions. If we orient a $\set{v_i,w_i}$ edge as $(v_i,w_i)$, we can only route the supply of $s_i$ via $w_i$ and $t_i$ to $t$, which requires $nT$ time units, due to the supply of $s_i$ and the capacity of $\set{w_i,t_i}$. If we orient all $\set{v_i,w_i}$ edges as $(w_i,v_i)$, we have to route the supply from $s_0$ via $v_0, v_1, \dots, v_k$ and $t_k$ to $t$, which requires $kT$ time units. By construction of the network, we have $k = (n-1)/4$, which proves the claimed factor.
 
 A similar construction can be employed in networks with a single source and multiple sinks (see Fig.~\ref{fig:timeboundsinks}). If we want to show the result for graphs $G$ that are trees, we can remove $t$ and shift the demand to the nodes $t_i$, $i=1,\dots,k$ and give node $t_i$ a demand of $-(nT)^{i-1}$.
 \begin{figure}[tb]
  \centering
  \begin{tikzpicture}[xscale=0.9]
   \node[fill,circle,minimum size=7mm] (s0) at (0,0) {\textcolor{white}{$s_0$}};
   \node[fill,circle,minimum size=7mm] (v0) at (0,2.5) {\textcolor{white}{$v_0$}}
    edge[thick,-] (s0);
   \node[fill,circle,minimum size=7mm] (v1) at (3,2.5) {\textcolor{white}{$v_1$}}
    edge[thick,-] node[auto,swap] {$\tau = T$} (v0);
   \node[fill,circle,minimum size=7mm] (w1) at (3,1.25) {\textcolor{white}{$w_1$}}
    edge[thick,-] (v1);
   \node[fill,circle,minimum size=7mm] (s1) at (2.5,0) {\textcolor{white}{$s_1$}}
    edge[thick,-] node[auto] {$u=(nT)^{k-2}$} (w1);
   \node[fill,circle,minimum size=7mm,label={-87:$b = -(nT)^{k-1}$}] (t1) at (3.5,0) {\textcolor{white}{$t_1$}}
    edge[thick,-] (w1);

   \node[fill,circle,minimum size=7mm] (v2) at (6,2.5) {\textcolor{white}{$v_2$}}
    edge[thick,-] node[auto,swap] {$\tau = T$} (v1);
   \node[fill,circle,minimum size=7mm] (w2) at (6,1.25) {\textcolor{white}{$w_2$}}
    edge[thick,-] (v2);
   \node[fill,circle,minimum size=7mm] (s2) at (5.5,0) {\textcolor{white}{$s_2$}}
    edge[thick,-] node[auto] {$u=(nT)^{k-3}$} (w2);
   \node[fill,circle,minimum size=7mm,label={south east:$b = -(nT)^{k-2}$}] (t2) at (6.5,0) {\textcolor{white}{$t_2$}}
    edge[thick,-] (w2);    
   
   \node (v3) at (8,2.5) {}
    edge[thick,-,dashed] node[auto,swap] {$\tau = T$} (v2);
   
   \node (w3) at (9,1.25) {$\dots$};
   
   \node (v4) at (10,2.5) {};
   
   \node[fill,circle,minimum size=7mm] (vk) at (12,2.5) {\textcolor{white}{$v_k$}}
    edge[thick,-,dashed] node[auto,swap] {$\tau = T$} (v4);
   \node[fill,circle,minimum size=7mm,label={south:$b = -(nT)^{0}$}] (tk) at (12,0) {\textcolor{white}{$t_k$}}
    edge[thick,-] (vk);      
    
   \node[fill,circle,minimum size=7mm,label={south:$b=\sum_{i=0}^{k-1} (nT)^i$}] (s) at (6.25,-1.75) {\textcolor{white}{$s$}};
   \draw[thick,-] (s0) |- (s);
   \draw[thick,-] (s1.265) -- ++(0,-1.0) -- (s);
   \draw[thick,-] (s2.265) -- ++(0,-1.0) -- (s);

  \end{tikzpicture}
  \caption{An undirected network with a single source where every orientation requires a time horizon that is larger by a factor of at least $(n-1)/4$ compared to the undirected setting. Not specified transit times and balances are 0 and not specified capacties are infinite.\label{fig:timeboundsinks}}
 \end{figure} 
 \qed
\end{proof}

A similar bound can be obtained for trees with unit capacities. Consider the instance depicted in Fig.~\ref{fig:tub}. In the undirected network, we can send the supply from a source $s_i$ to the sink $t_i$ within the time horizon of $T+1$. In any orientation, we have to use the supply of a source $s_i$, $1 < i \leq k$ to fulfill the demand of sink $t_{i-1}$. This forces us to use the supply of $s_1$ to fulfill the demand of $t_k$, which takes at least $kT+1$ time units.

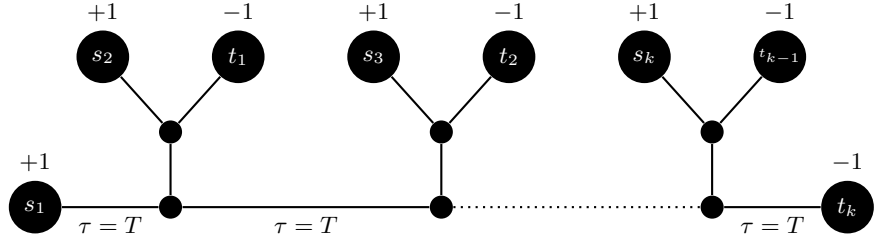
\begin{figure}[tb]
 \centering
 \begin{tikzpicture}[xscale=0.9]
  \node (d1) at (0,0)[fill,circle,minimum size=7mm,label={90: $+1$}] {\textcolor{white}{$s_1$}};
  \node (d2) at (2,0)[fill,circle] {};
  \node (d3) at (2,1)[fill,circle] {};
  \node (d4) at (3,2)[fill,circle,minimum size=7mm,label={90: $-1$}] {\textcolor{white}{$t_1$}};
  \node (d5) at (1,2)[fill,circle,minimum size=7mm,label={90: $+1$}] {\textcolor{white}{$s_2$}};
  \node (d6) at (6,0)[fill,circle,minimum size=0mm] {};
  \node (d7) at (6,1)[fill,circle,minimum size=0mm] {};
  \node (d8) at (7,2)[fill,circle,minimum size=7mm,label={90: $-1$}] {\textcolor{white}{$t_2$}};
  \node (d9) at (5,2)[fill,circle,minimum size=7mm,label={90: $+1$}] {\textcolor{white}{$s_3$}};  
  \node (d10) at (10,0)[fill,circle,minimum size=0mm] {};
  \node (d11) at (10,1)[fill,circle,minimum size=0mm] {};
  \node (d12) at (11,2)[fill,circle,minimum size=7mm,inner sep=0mm,label={90: $-1$}] {\textcolor{white}{\tiny $t_{k-1}$}};
  \node (d13) at (9,2)[fill,circle,minimum size=7mm,label={90: $+1$}] {\textcolor{white}{$s_k$}};    
  \node (d14) at (12,0)[fill,circle,minimum size=7mm,label={90: $-1$}] {\textcolor{white}{$t_k$}};
  \draw [thick](d1) -- node[auto,swap] {$\tau = T$} (d2);
  \draw [thick](d2) -- (d3);
  \draw [thick](d3) -- (d4);
  \draw [thick](d5) -- (d3);
  \draw [thick](d2) -- node[auto,swap] {$\tau = T$} (d6);
  \draw [thick](d6) -- (d7);
  \draw [thick](d7) -- (d8);
  \draw [thick](d9) -- (d7); 
  \draw [thick, dotted](d6) -- (d10);
  \draw [thick](d10) -- (d11);
  \draw [thick](d11) -- (d12);
  \draw [thick](d13) -- (d11);
  \draw [thick](d10) -- node[auto,swap] {$\tau = T$} (d14);
 \end{tikzpicture}
 \caption{An undirected network with unit capacities where all supplies and demands can be fulfilled within a time horizon of $T+1$. However, any orientation requires a time horizon of at least $kT+1$. Not specified transit times and balances are 0.\label{fig:tub}}
\end{figure}

\subsection{Price in Terms of Flow and Time Horizon}

We have seen now that the price of orientation is 3 with regard to the flow value, and $\Omega(n)$ with regard to the time horizon. We can improve on these bounds if we allow to pay the price of orientation partly in terms of flow value and partly in terms of the time horizon. This is possible by combining the reduction to maximum flows over time from Theorem~\ref{theorem:flowpriceupperbound} with the concept of temporally averaged flows (see , \eg, \cite{FLSK07}).

\begin{theorem}
 \label{theorem:bicriteria}
 Let $N = (G, u, b, \tau, T)$ be an undirected network over time, in which $B$ units of flow can be sent within the time horizon $T$. Then there exists an orientation $\orient{N}$ of $N$ in which at least $B/2$ units of flow can be sent within time horizon $2T$. The orientation and a transshipment over time with this property can be obtained in polynomial time.
\end{theorem} 

\begin{proof}
 In order to prove this claim, we will create a modified network with a larger time horizon in which we can send a temporally repeated flow which uses each edge in only one direction. This gives us then an orientation with the desired properties. 
 Consider the network $N' = (G', u', b', \tau', 2T)$ defined by 
 \ama
  V(G') &:= V(G) \cup \set{s,t},\\ \quad 
  E(G') &:= E(G) \cup \setc{\set{s,v}}{b_v > 0} \cup \setc{\set{v,t}}{b_v < 0},\\
  u'_e  &:= \begin{cases} \frac{b_v}{T} & e = \set{s,v} \\ \frac{-b_v}{T} & e = \set{v,t} \\ u_e & \text{else}\end{cases}, \quad
  \tau'_e := \begin{cases} \tau_e & e \in E(G) \\ 0 & \text{else} \end{cases}, \\
  b'_v  &:= \begin{cases} 0 & v \in V(G) \\ B & v = s \\ -B & v = t \end{cases}.\\
 \ema 
 An illustration can be found in Fig.~\ref{fig:bicriteria}.  
 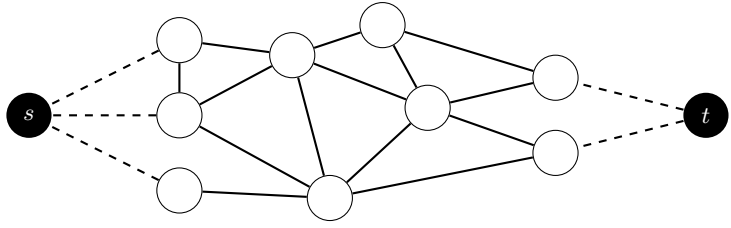
\begin{figure}[tb]
  \centering
  \begin{tikzpicture}
   \node[fill,circle,minimum size=6mm] (s) at (0,0) {\textcolor{white}{$s$}};
   \node[draw,circle,minimum size=6mm] (a) at (2,-1) {}
    edge[thick,-,dashed] (s);
   \node[draw,circle,minimum size=6mm] (b) at (2,-0) {}
    edge[thick,-,dashed] (s);
   \node[draw,circle,minimum size=6mm] (c) at (2,1) {}
    edge[thick,-] (b) 
    edge[thick,-,dashed] (s);
   \node[draw,circle,minimum size=6mm] (d) at (3.5,0.8) {}
    edge[thick,-] (b)
    edge[thick,-] (c);
   \node[draw,circle,minimum size=6mm] (e) at (4.0,-1.1) {}
    edge[thick,-] (a)
    edge[thick,-] (b)
    edge[thick,-] (d);
   \node[draw,circle,minimum size=6mm] (f) at (4.7,1.2) {}
    edge[thick,-] (d);
   \node[draw,circle,minimum size=6mm] (g) at (5.3,0.1) {}
    edge[thick,-] (d)
    edge[thick,-] (e)
    edge[thick,-] (f);
   \node[draw,circle,minimum size=6mm] (h) at (7.0,0.5) {}
    edge[thick,-] (f)
    edge[thick,-] (g);
   \node[draw,circle,minimum size=6mm] (i) at (7.0,-0.5) {}
    edge[thick,-] (e)
    edge[thick,-] (g);
   \node[fill,circle,minimum size=6mm] (t1) at (9,0) {\textcolor{white}{$t$}}
    edge[thick,-,dashed] (h)
    edge[thick,-,dashed] (i);
  \end{tikzpicture}
  \caption{The modified network consisting of the original network (white) and the newly introduced nodes (black) and auxiliary edges (dashed). \label{fig:bicriteria}}
 \end{figure}
 We know that there is a transshipment over time $f$ that sends $B$ flow units within time $T$ in $N$. We can decompose this transshipment into flow along a family of paths $\mathcal{P}$ with $\tau_P < T$ for all $P \in \mathcal{P}$ and interpret $f$ as sending flow into paths $P \in \mathcal{P}$ at a rate of $f_P(\theta)$ at time $\theta$. Now consider a transshipment over time $f'$ that is defined by sending flow into the same paths as $f$, but at an averaged rate of $f'_P(\theta) := \frac{1}{T}\int_{0}^T f_P(\xi)\ d\xi  $ for a path $P$ and a time $\theta \in [0,T)$. Since all paths $P \in \mathcal{P}$ have $\tau_P < T$, $f'$ sends its flow within a time horizon of $2T$. $f'$ sends $B$ flow units as well, since we just averaged flow rates and the averaging guarantees that the capacities of the edges $e \in E(G') \setminus E(G)$ are not violated. We conclude that a maximum flow over time in $N'' := (G', u', \tau', s, t, 2T)$ has a value of at least $B$.
 
 Now we compute a maximum flow over time in $N''$ using the Ford-Fulkerson algorithm \cite{FoFu:62}. This algorithm computes a temporally repeated maximum flow over time $f''$ which uses each edge in only one direction. We can transform $f''$ into a transshipment over time $f^*$ for $N$ by cutting off the edges of $E(G') \setminus E(G)$. Due to $u'_{\set{s,v}} = b_v/T$, $u'_{\set{v,t}} = -b_v/T$ and the time horizon of $2T$, the resulting flow over time $f^*$ satisfies supplies and demands $b''$ with $0 \leq b''_v \leq 2b_v$ for $v \in V(G)$ with $b_v > 0$ and  $0 \geq b''_v \geq 2b_v$ for $v \in V(G)$ with $b_v < 0$. Thus, $1/2 f^*$ sends at least $B/2$ flow units in $2T$ time and uses each edge in only one direction without violating the balances $b$. Furthermore, this can be done in polynomial time, since the transformation and the Ford-Fulkerson algorithm are polynomial. This concludes the proof. \qed
\end{proof}

\paragraph{\textbf{Earliest Arrival Flows.}}

We now have tight bounds for the flow and time price of orientation for maximum or quickest flows over time. However, for application in evacuations, it would be nice if we could analyze the price of orientation for so-called earliest arrival flows as well, as they provide guarantees for flow being sent at all points in time. Unfortunately, we can create instances where not even approximate earliest arrival contraflows exist, because the trade-off between different orientations becomes too high. 

Earliest arrival flows are special quickest flows that maximize the number of flow units that have reached a sink at each point in time simultaneously. This is an objective that is very desirable in evacuation management, if the exact amount of available time is not clear in the planning stage. It is not clear that these flows exist in general, and indeed their existence depends on the number of sinks. Earliest arrival flows always exist if only one sink is present, as was first proven by Gale~\cite{GAL59}. For multiple sinks, that is usually not the case, but approximations are still possible~\cite{BaumannKoehler07,GrKaScSc:12}.

For every time $\theta \in \R^+$, let $f^*_\theta$ be a maximum flow over time with time horizon $\theta$. We define $p(\theta) := |f^*_\theta|_\theta$ and refer to the values $p(\theta)$ as the \emph{earliest arrival pattern}. An \emph{earliest arrival flow} is a flow over time $f$ which simultaneously satisfies $|f|_{\theta} = p(\theta)$ for all points in time $\theta \in [0,T)$, respectively. 

An \emph{$\alpha$-time-ap\-prox\-imate earliest arrival flow} is a flow over time $f$ that achieves at every point in time $\theta \in [0,T)$, respectively, at least as much flow value as possible at time $\theta/\alpha$, \ie, $|f|_{\theta} \geq p\left(\frac{\theta}{\alpha}\right)$. A \emph{$\beta$-value-approximate earliest arrival flow} is a flow over time $f$ that achieves at every point in time $\theta \in [0,T)$, respectively, at least a $\beta$-fraction of the maximum flow value at time $\theta$, \ie, $|f|_{\theta} \geq \frac{p(\theta)}{\beta}$.

In practice, orienting road networks is an important aspect of evacuation management. In terms of evacuations, earliest arrival flows (or approximations of them) are very desirable, as they provide optimal routings independent of the time that is available. The contraflow versions of these problems ask for an orientation $\orient{N}$ of $N$ and a flow over time $f$ in $\orient{N}$, such that $|f|_\theta = p(\theta)$, $|f|_{\theta} \geq p\left(\frac{\theta}{\alpha}\right)$ and $|f|_{\theta} \geq \frac{p(\theta)}{\beta}$, respectively, for all $\theta$. Notice that $p$ refers to the earliest arrival pattern of the undirected network in this case. 

We are able to show that earliest arrival flows and the approximations developed in~\cite{BaumannKoehler07,GrKaScSc:12} do not exist in this setting.

\begin{theorem}
 \label{theorem:lowerboundeaf}
 There are undirected networks over time $N = (G, u, b, \tau)$ for which an earliest arrival flow exists, but that do not allow for an earliest arrival contraflow. This also holds for $\alpha$-time- and $\beta$-value-approximative earliest arrival contraflows for $\alpha < T/2$ and $\beta < U$, where $T$ and $U$ are the largest transit time and capacity in the network.
\end{theorem}

\begin{proof}
 Consider the network depicted in Fig.~\ref{fig:eafexample}. 
 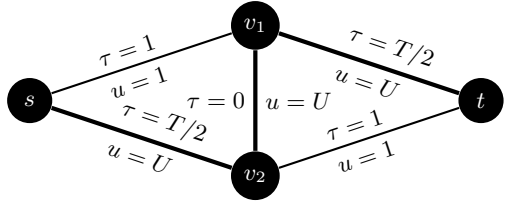
\begin{figure}[tb]
  \centering
  \begin{tikzpicture}
   \node[fill,circle,minimum size=6mm] (a) at (0,1) {\textcolor{white}{$s$}};
   \node[fill,circle,minimum size=6mm] (b1) at (3,0) {\textcolor{white}{$v_2$}}
    edge[ultra thick,-] node[auto,sloped,pos=0.25] {$u = U$} node[auto,swap,sloped,pos=0.70] {$\tau = T/2$} (a);
   \node[fill,circle,minimum size=6mm] (b2) at (3,2) {\textcolor{white}{$v_1$}}
    edge[thick,-] node[auto,sloped,pos=0.75] {$u = 1$} node[auto,swap,sloped,pos=0.35] {$\tau = 1$} (a)
    edge[ultra thick,-] node[auto] {$u = U$} node[auto,swap] {$\tau = 0$} (b1);
   \node[fill,circle,minimum size=6mm] (a) at (6,1) {\textcolor{white}{$t$}}
    edge[thick,-] node[auto,sloped,pos=0.75] {$u = 1$} node[auto,swap,sloped,pos=0.35] {$\tau = 1$} (b1)
    edge[ultra thick,-] node[auto,sloped,pos=0.25] {$u = U$} node[auto,swap,sloped,pos=0.70] {$\tau = T/2$} (b2);
  \end{tikzpicture}
  \caption{An undirected network with source $s$ and sink $t$ and capacities and transit times as specified. \label{fig:eafexample}}
 \end{figure}
 We can orient the edge $\set{v_1,v_2}$ as $(v_2,v_1)$ and have flow arriving with a rate of $U\eps$ starting at time $2T$. However, we have no flow arriving before time $T+1$ using this orientation. If we use the orientation $(v_1,v_2)$ instead, we can have flow arrive at time 2, but at a rate of 1 instead of $U$. For $U \gg T$, this trade-off makes it impossible to find an earliest arrival contraflow. 
 
 For $\alpha$-time-approximative earliest arrival flows, we can choose $U = 2T^2+T$. This yields an instance where no $\alpha < T$-approximation is possible. Sending flow at a rate of 1 using the orientation $(v_1,v_2)$ results in $U$ flow units being sent until time $2T^2+T+2$ (or with a rate of $2$, it takes $T^2+3/2T+1$). However, using the orientation $(v_2,v_1)$, we could have sent them by time $2T+1$. Sending flow at a rate of $U$ using the orientation $(v_2,v_1)$ results in no flow units being sent until time $T+1$, but flow could have been sent as early as time $2$ using the other orientation. This yields the non-approximability result for $\alpha$-time-approximations. 
 
 For $\beta$-value-approximations, we need to use the orientation $(v_1,v_2)$ to have some flow arrive starting at time $2$. However, using the other orientation allows us to send flow at a rate of $U$, yielding a ratio that converges to $U$, which concludes the proof.
 \qed
\end{proof}
\section{Complexity Results}
\label{sec:hardness}
 Furthermore, we can show non-approximability results for several contraflow over time problems. More specifically, we can show that neither quickest contraflows nor maximum contraflows over time can be approximated better than a factor of $2$, unless $P=NP$. For multicommodity contraflows over time, we can even show that maximum multicommodity concurrent contraflows and quickest multicommodity contraflows cannot be approximated at all, even with zero transit times, unless $P = NP$. 
 
 \begin{theorem}
\label{thm:hardness:quickest}
 The quickest contraflow problem cannot be approximated better than a factor of $2$, unless $P=NP$.
\end{theorem}

\begin{proof}
 Rebennack et. al~\cite{RebEtAl:10} showed the NP-hardness of this problem. The reduction technique they provide can also be used to show a non-approximability claim, if we modify the transit times used in their reduction. 
 We give a brief sketch of their reduction technique, which is based on the SAT problem. We construct an instance for the quickest contraflow problem from an instance for the 3-SAT problem with $\ell$ clauses $c_1,\dots,c_\ell$ over $k$ variables $x_1,\dots,x_k$ as follows.
 \begin{enumerate}
  \item For each clause $c_i$, we create a source $c_1^+$ and a sink $c_1^-$ with a supply and demand of 1 and -1, respectively. 
  \item For each variable $x_i$, we create four nodes: $x^1_i$ and $x^2_i$ for its unnegated literal, and $\bar{x}^1_i$, $\bar{x}^2_i$ for its negated literal. These nodes get neither supplies nor demands. Furthermore, we create a source $s_i$ and a sink $t_i$ with a supply and demand of 1, respectively. Finally, we create edges $\set{x^1_i,x^2_i}$, $\set{\bar{x}^1_i,\bar{x}^2_i}$, $\set{s_i,x^2_i}$, $\set{s_i,\bar{x}^2_i}$ with a transit time of $\tau_2$ and edges $\set{t_i,x^1_i}$, $\set{t_i,\bar{x}^1_i}$ with a transit time of $\tau_1$.
  \item For each clause $c_i = x_{i_1} \vee x_{i_2} \vee \bar{x}_{i_3}$ we create edges $\set{c_i^+,x_{i_1}^1}$, $\set{c_i^+,x_{i_2}^1}$, $\set{c_i^+,\bar{x}_{i_3}^1}$ with a transit time of $\tau_1$ and edges $\set{c_i^-,x_{i_1}^2},\set{c_i^-,x_{i_2}^2},\set{c_i^-,\bar{x}_{i_3}^2}$ with a transit time of $\tau_2$.
 \end{enumerate}
 All capacities are infinite. Fig.~\ref{fig:satRed} depicts such a construction. 

 \paragraph{YES-Instance $\to$ Routable in time $\tau_1+2\tau_2$.}
 We derive an orientation from an assignment for the SAT problem that fulfills all clauses. If variable $x_i$ is set to $1$ in the assignment, we orient $\set{x^1_i,x^2_i}$ as $(x^1_i,x^2_i)$ and $\set{\bar{x}^1_i,\bar{x}^2_i}$ as $(\bar{x}^2_i,\bar{x}^1_i)$. Otherwise, we orient $\set{x^1_i,x^2_i}$ as $(x^2_i,x^1_i)$ and $\set{\bar{x}^1_i,\bar{x}^2_i}$ as $(\bar{x}^1_i,\bar{x}^2_i)$. All other edges are oriented away from the sources or towards the sinks, respectively. A clause source $c_i^+$ with a fulfilled literal $x_{i}$ can send 1 flow unit along $c_i^+ \to x_i^1 \to x_i^2 \to c_i^-$, and each variable source $s_i$ can send 1 flow unit to its sink via $s_i \to x_i^2 \to x_i^1 \to t_i$ if $x_i = 0$ and $s_i \to \bar{x}_i^2 \to \bar{x}_i^1 \to t_i$ otherwise. This takes $\tau_1+2\tau_2$ time units.

 \paragraph{NO-Instance $\to$ Not routable in time $< 2\tau_1$.}
 We set $t_2 = 0$, as above. If we want to send everything in a time $< 2\tau_1$, we can only use paths containing at most one $\tau_1$ edge. It follows that supply from the clause sources needs to go to a clause sink, via an $(x^1_i,x^2_i)$ or $(\bar{x}^1_i,\bar{x}^2_i)$ edge. Similar, each variable sink $t_i$ needs to get its flow from a variable source and requires an $(x^2_i,x^1_i)$ or $(\bar{x}^2_i,\bar{x}^1_i)$ oriented edge, if we want to be faster than $2\tau_1$. Having both edges oriented as $(x^2_i,x^1_i)$ and $(\bar{x}^2_i,\bar{x}^1_i)$ does not help more than having only one of them oriented that way -- we will now assume without loss of generality, that only one of the edges is oriented that way. We can derive an assignment from the orientation of these edges. If we have $(x^2_i,x^1_i)$ in our orientation, we set $x_i = 0$ and $x_i = 1$ otherwise. However, no assignment fulfills all clauses, therefore we have to send clause supplies to variable demands, which takes $2\tau_1$.
 
 Thus, if we are able to approximate the quickest contraflow problem within a factor of $\frac{2\tau_1}{\tau_1+2\tau_2}$, then we can distinguish between YES and NO instances of the 3-SAT problem. For $\tau_2 = 0$, this yields the result. \qed
\begin{figure}[tb]
 \begin{center}
       \begin{tikzpicture} [>=stealth]
	\node (x11) at (4,6)   [blacknode] {$x^1_1$};	
	\node (x21) at (5.5,6)   [blacknode] {$x^2_1$};	
	\node (x'11) at (4,5)   [blacknode] {$\bar{x}^1_1$};	
	\node (x'21) at (5.5,5)   [blacknode] {$\bar{x}^2_1$};	
	\node (x12) at (4,4)   [blacknode] {$x^1_2$};	
	\node (x22) at (5.5,4)   [blacknode] {$x^2_2$};	
	\node (x'12) at (4,3)   [blacknode] {$\bar{x}^1_2$};	
	\node (x'22) at (5.5,3)   [blacknode] {$\bar{x}^2_2$};	
	\node (x1n) at (4,1.5)   [blacknode] {$x^1_k$};	
	\node (x2n) at (5.5,1.5)   [blacknode] {$x^2_k$};	
	\node (x'1n) at (4,0.5)   [blacknode] {$\bar{x}^1_k$};	
	\node (x'2n) at (5.5,0.5)   [blacknode] {$\bar{x}^2_k$};	
	
	\node (c1) at (1,5.5)   [blacknode,label={west: $+1$}] {$c^+_1$};	
	\node (c2) at (1,4)   [blacknode,label={west: $+1$}] {$c^+_2$};	
	\node (cm) at (1,1)   [blacknode,label={west: $+1$}] {$c^+_\ell$};	
	\node (c'1) at (8,5.5)   [blacknode,label={east: $-1$}] {$c^-_1$};	
	\node (c'2) at (8,4)   [blacknode,label={east: $-1$}] {$c^-_2$};	
	\node (c'm) at (8,1)   [blacknode,label={east: $-1$}] {$c^-_\ell$};	

	\node (xt1) at (3,7)   [blacknode,label={west: $-1$}] {$t_1$};	
	\node (xs1) at (6.5,7)   [blacknode,label={east: $+1$}] {$s_1$};	
	\node (xtn) at (3,-0.5)   [blacknode,label={west: $-1$}] {$t_k$};	
	\node (xsn) at (6.5,-0.5)   [blacknode,label={east: $+1$}] {$s_k$};	

    	\draw [thick,dashed](c1) -- (x11)   ;
	\draw [thick,dashed](c1) -- (x1n)   ;
	\draw [thick,dashed](c1) -- (x'12)  ;
	\draw [thick,dashed](c2) -- (x'11)  ;
	\draw [thick,dashed](c2) -- (x1n)  ;
	\draw [thick,dashed](cm) -- (x'1n)  ;	

	\draw [thick,dashed](xt1) -- (x11)  ;		
	\draw [thick,dashed](xt1) -- (x'11)  ;		
	\draw [thick,dashed](xtn) -- (x1n)  ;		
	\draw [thick,dashed](xtn) -- (x'1n)  ;		

   	\draw [thick](c'1) -- (x21)   ;
	\draw [thick](c'1) -- (x2n)   ;
	\draw [thick](c'1) -- (x'22)  ;
	\draw [thick](c'2) -- (x'21)  ;
	\draw [thick](c'2) -- (x2n)  ;
	\draw [thick](c'm) -- (x'2n)  ;	

	\draw [thick](xs1) -- (x21)  ;		
	\draw [thick](xs1) -- (x'21)  ;		
	\draw [thick](xsn) -- (x2n)  ;		
	\draw [thick](xsn) -- (x'2n)  ;
	
	\draw [thick](x11) -- (x21);	
	\draw [thick](x'11) -- (x'21);
	\draw [thick](x12) -- (x22);
	\draw [thick](x'12) -- (x'22);
	\draw [thick](x1n) -- (x2n);
	\draw [thick](x'1n) -- (x'2n);
	
	\node (dots) at (1,2.75) [scale = 1.5]{\textbf{$\vdots$}};	
	\node (dots) at (8,2.75) [scale = 1.5]{\textbf{$\vdots$}};	
	\node (dots) at (4.75,2.35) [scale = 1]{\textbf{$\vdots$}};	
  \end{tikzpicture}
 \end{center}
 \caption{The quickest contraflow instance derived from the SAT instance. Edges with a transit time of $\tau_1$ are dashed, edges with a transit time of $\tau_2$ are solid. $s_2$ and $t_2$ and several clause-edges are not shown. \label{fig:satRed}}
\end{figure}
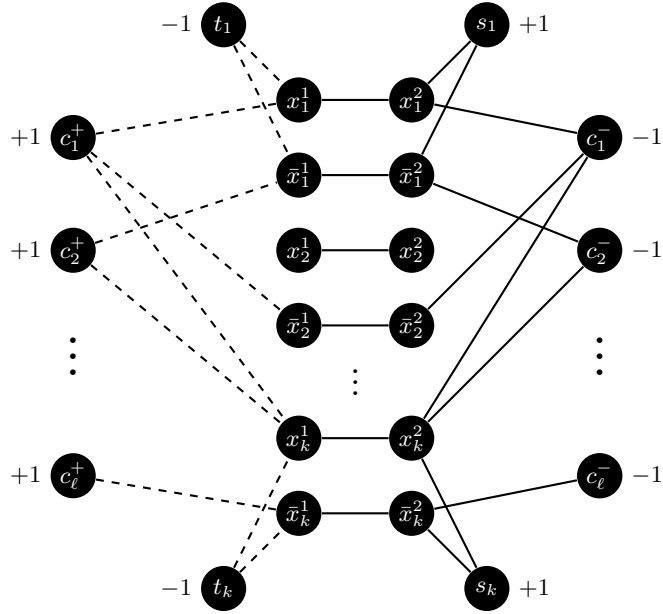 

\end{proof}

\begin{theorem}
 \label{thm:hardness:maximum}
 The maximum contraflow over time problem cannot be approximated better than a factor of $2$, unless $P = NP$.
\end{theorem}

\begin{proof}
 The following reduction is inspired by \cite{KlinzWoeginger04}. Consider an instance of the \textsc{PARTITION}-problem, given by integers $a_1$, $\dots$, $a_n$ with $\sum_{i=1}^n a_i = 2L$ for some integer $L > 0$. We create an instance for the maximum contraflow over time problem as follows:
 \begin{enumerate}
  \item We create $n+1$ nodes $v_1,\dots,v_{n+1}$, two sources $s_1,s_2$ and two terminals $t_1,t_2$. The sources have each a supply of 1, the sinks each a demand of $-1$.
  \item We create $2(n+1)$ edges $e_i = \set{v_i,v_{i+1}}$, $e'_i = \set{v_i,v_{i+1}}$ with a transit time of $a_i$ for $e_i$ and a transit time of $0$ for $e'_i$, edges $\set{s_1,v_1}$, $\set{t_2,v_1}$ with a transit times of $L+1$ and edges $\set{s_2,v_{n+1}}$, $\set{t_1,v_{n+2}}$ with a transit times of $0$. All edges have unit capacities.
  \item We set the time horizon to $2L+2$.
 \end{enumerate}
 The resulting instance is depicted in Fig.~\ref{fig:mc:max}. 

 \begin{figure}[tb]
  \centering
  \begin{tikzpicture}[xscale=0.9]
   \foreach \n / \i / \x in {1/1/0.0,2/2/2.0,3/3/4.0,4/n/6.0,5/n+1/8.0}{
    \node[node,fill=black,text=white] (v\n) at (\x,0) {$v_{\i}$};
   }   
   \foreach \n / \i / \o in {1/1/2,2/2/3,4/n/5}{
    \draw[edge, bend left = 45] (v\n.45) ..controls +(0.3,0.3) and +(-0.4,0.3).. node[auto] {$a_\i$} (v\o.135);
    \draw[edge, bend left = 45] (v\n.315) ..controls +(0.3,-0.3) and +(-0.4,-0.3).. node[auto,below] {$0$} (v\o.225);
   }
   \node[node,fill=black,text=white,label={north:$1$}] (s1) at (-2,1) {$s_{1}$};
   \node[node,fill=black,text=white,label={north:$1$}] (s2) at (10,-1) {$s_{2}$};
   \node[node,fill=black,text=white,label={north:$-1$}] (t1) at (10,1) {$t_{1}$};
   \node[node,fill=black,text=white,label={north:$-1$}] (t2) at (-2,-1) {$t_{2}$};
   \draw[edge] (s1) -- (v1) node[auto] {$L+1$};
   \draw[edge] (t2) -- (v1) node[auto] {$L+1$};
   \draw[edge] (s2) -- (v5) node[auto] {$0$};
   \draw[edge] (t1) -- (v5) node[auto] {$0$};   
   \node[] (dots) at (5.0,0.0) {$\dots$}; 
  \end{tikzpicture}
  \caption{The maximum contraflow over time problem instance.\label{fig:mc:max}}
 \end{figure}
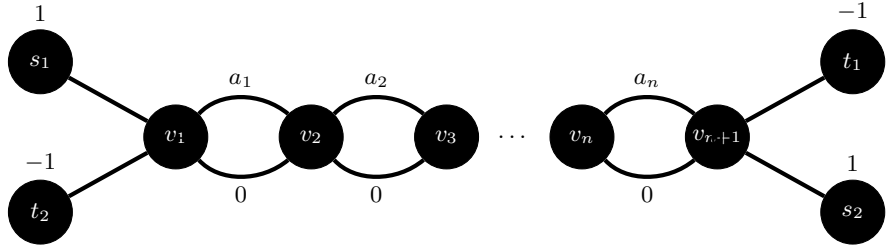

 \paragraph{YES-Instance $\to$ total flow value of 2.}
 If the PARTITION-instance is a YES-instance, there exists a subset of indices $I \subseteq \set{1,\dots,n}$ such that $\sum_{i\in I} A_i = L$. Thus, we can orient the edges so that two disjoint $v_1$-$v_{n+1}$- and $v_{n+1}$-$v_{1}$-paths with a length of $L$ are created, which gives us two disjoint paths from $s_1$ to $t_1$ and $s_2$ to $t_2$ with transit time $2L+1$, respectively, that are sufficient to send all supplies within the time horizon of $2L+2$. 

 \paragraph{NO-Instance $\to$ total flow value of 1.}
 Notice that we cannot send any flow from $s_1$ to $t_2$ within the time horizon of $2L+2$. If the PARTITION-instance is a NO-instance, then we cannot get a $v_1$-$v_{n+1}$- and a $v_{n+1}$-$v_1$-path with a length $L$, so only the demands of one of the two commodities can be fulfilled. \qed
\end{proof}

\paragraph{\textbf{Multicommodity Flows over Time.}} Maximum flows over time and quickest flows can also be defined for the case of \emph{multiple commodities}. In this case we replace the supplies and demands $b$ by supplies and demands $b^i$ for all commodities $i = 1,\dots,k$. Each commodity has to fulfill its own flow conservation constraints, and supply from one commodity can only be used for the demands of the same commodity. However, the capacities of the network are shared by all commodities.
This generalization leads to \emph{maximum multicommodity (contra)flow over time}, \emph{quickest muticommodity (contra)flow} problems. In this setting it can also be interesting to maximize the minimal fraction of flow of each commodity to its total demand. This is referred to as \emph{concurrent multicommodity (contra)flow over time} problem.
\begin{theorem}
 \label{thm:hardness:mc:concurrent}
 Unless $P = NP$, the maximum multicommodity concurrent contraflow prob\-lem over time cannot be approx\-i\-ma\-ted by time or value. This holds even in the case with zero transit times.
\end{theorem}

\begin{proof}
 Consider an instance of 3-SAT, given by a set of $\ell$ clauses $C = \set{c_1,\dots,c_\ell}$ on $k$ variables $x_1, \dots, x_k$. We create a corresponding instance of the maximum concurrent contraflow problem as follows:
 \begin{enumerate}
  \item For each clause $c_i$ we create a node $c_i$, 
  \item for each variable $x_i$, create nodes $x_i^1$, $x_i^2$, $\overline{x}_i^1$, $\overline{x}_i^2$,$x_i^-$,$\overline{x}_i^-$, $d_i^-$, $\overline{d}_i^-$ and $d_i^+$.
  \item For a clause $c_i = x_{i_1} \vee x_{i_2} \vee \overline{x}_{i_3}$ we create edges $\set{c_i,x_{i_1}^1}$, $\set{c_i,x_{i_2}^1}$ and $\set{c_i,\overline{x}_{i_3}^1}$,
  \item for each variable $x_i$, create edges $\set{d_i^+,x_i^2}$, $\set{d_i^+,\overline{x}_i^2}$, $\set{x_i^1,d_i^-}$, $\set{\overline{x}_i^1,\overline{d}_i^-}$, $\set{x_i^2,x_i^-}$, $\set{\overline{x}_i^2,\overline{x}_i^-}$, $\set{x_i^1,x_i^2}$ and $\set{\overline{x}_i^1,\overline{x}_i^2}$. 
  \item Capacities are set to $\ell$ for each edge and transit times to $0$.
  \item There is a commodity for each variable $x_i$, with a supply of 2 at $d_i^+$ and demands of $-1$ at $d_i^-$, $\overline{d}_i^-$. Furthermore, there is a commodity for each clause $c_i = x_{i_1} \vee x_{i_2} \vee \overline{x}_{i_3}$, with a supply of 3 at the clause node $c_i$ and a demand of $-1$ at $x_{i_1}^-$, $x_{i_2}^-$ and $\overline{x}_{i_3}^-$.
 \end{enumerate}
 Notice that the resulting network -- an example of which is depicted in Fig.~\ref{figure2} -- has $\ell + 9k$ nodes and $3\ell + 8k$ edges, which is polynomial in the size of the 3-SAT instance.

 \begin{figure}[tb]
  \centering
  \begin{tikzpicture}[xscale=0.8]
   \splitvariableblock{0}{0}{1}
   \splitvariableblock{4.75}{0}{2}
   \node (dots) at (7.5,-1.5) {$\dots$};
   \splitvariableblock{10.25}{0}{k};
   \node[draw, rounded rectangle, fill=white,label={[label distance=-0mm]90:\tiny $3$}] (cl1) at (3.5,1.5) {$c_1 = \overline{x}_1 \vee x_2 \vee \overline{x}_k$};   
   \draw[longEdge] (cl1.200) -- ++(0,-0.375) -- (b1.45);
   \draw[longEdge] (cl1.south) -- ++(0,-0.375) -| (a2.north);
   \draw[longEdge] (cl1.east) -| (bk.north);
   \node[draw, rounded rectangle, fill=white,label={[label distance=-0mm]90:\tiny $3$}] (cl2) at (6.0,2.25) {$c_2 = \overline{x}_1 \vee \overline{x}_2 \vee x_k$};   
   \draw[longEdge] (cl2.west) -| (b1.north);   
   \draw[longEdge] (cl2.south) -- +(0,-0.75) -| (b2.north);   
   \draw[longEdge] (cl2.east) -| (ak.north);
   \node[draw, rounded rectangle, fill=white,label={[label distance=-0mm]90:\tiny $3$}] (clk) at (11.0,2.25) {$c_\ell = \dots$};
  \end{tikzpicture}
  \caption{The maximum multicommodity concurrent flow problem instance.\label{figure2}}
 \end{figure}
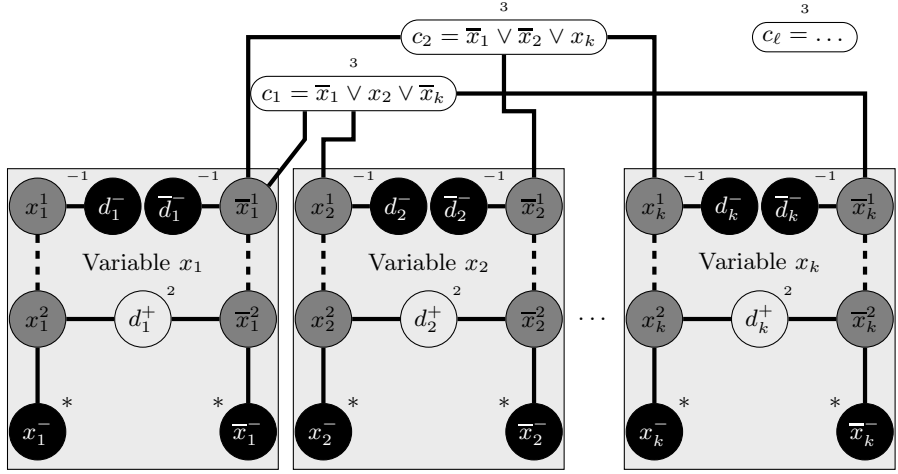

 \paragraph{YES-Instance $\to$ $\frac{1}{3}$-concurrent flow value.}

 If the 3-SAT instance is a YES-instance, then there is a variable assignment $x_i \in \set{0,1}$, $i = 1,\dots,k$ fulfilling all clauses. We use this assignment to define an orientation of the edges in our network. If a variable $x_i$ is assigned a value of 0, we orient the edge $\set{x_i^1,x_i^2}$ as $(x_i^2,x_i^1)$ and $\set{\overline{x}_i^1,\overline{x}_i^2}$ as $(\overline{x}_i^1,\overline{x}_i^2)$; if $x_i$ is assigned the value $1$, we orient edge $\set{\overline{x}_i^1,\overline{x}_i^2}$ as $(\overline{x}_i^2,\overline{x}_i^1)$ and $\set{x_i^1,x_i^2}$ as $(x_i^1,x_i^2)$. All other edges are oriented away from sources and towards sinks.
 Notice that:
 \begin{enumerate}
  \item Each clause commodity of a clause $c_i = x_{i_1} \vee x_{i_2} \vee \overline{x}_{i_3}$ can send flow to the nodes $x^1_{i_1}, x^1_{i_2}, \overline{x}^1_{i_3}$.
  \item Each clause is satisfied by our assignment, so there is a literal in each clause that is true.
  \item For this literal $x_i$, there is an edge directed from $x_i^1$ to $x_i^2$  (or $\overline{x}_i^1$ to $\overline{x}_i^2$, respectively).
  \item By construction of the instance, there is a demand for this clause commodity in $x_i^-$, which can be reached from $x_i^2$ (or $\bar{x}_i^-$ and $\bar{x}_i^2$, respectively).
 \end{enumerate}
 Therefore we can fulfill as much demand of a clause commodity as it has satisfied literals in our assignment, which is at least 1. Thus, we have a concurrent flow value of $\frac{1}{3}$ for these commodities. Now we need to consider the variable commodities. Since our assignment can only set $x_i$ to either 1 or 0, one of the edges $\set{x_i^1,x_i^2}$ and $\set{\overline{x}_i^1,\overline{x}_i^2}$ has been oriented as $(x_i^2,x_i^1)$ or $(\overline{x}_i^2,\overline{x}_i^1)$, respectively, in each variable block. This creates a path to send one flow unit of each variable commodity, giving us a concurrent flow value of $\frac{1}{2}$ for them, yielding a total concurrent flow value of $\frac{1}{3}$.

 \paragraph{Positive concurrent flow value $\to$ YES-Instance.}
 
 In order to have a positive concurrent flow value, at least one of the edges $\set{x_i^1,x_i^2}$ and $\set{\overline{x}_i^1,\overline{x}_i^2}$ needs to be oriented as $(x_i^2,x_i^1)$ or $(\overline{x}_i^2,\overline{x}_i^1)$ in each variable block -- otherwise there is no way to route any flow from the variable commodity. Thus, we can define an assignment by setting $x_i = 0$ if $\set{x_i^1,x_i^2}$ has been oriented as $(x_i^2,x_i^1)$ and $x_i = 1$ otherwise. Notice that if $\set{x_i^1,x_i^2}$ is oriented as $(x_i^2,x_i^1)$, there is no flow reaching $x_i^-$ (the edges adjacent to $d_i^+$ cannot be used to reach $x_i^-$, or no flow of the variable commodity could be sent). But since we have a positive concurrent flow value, there is flow from every clause commodity reaching one of its sinks. Such flow has -- by construction of the network -- to travel through the block of one of the variables contained in the clause. More specifically, it has to traverse the $(x_i^1,x_i^2)$ or $(\overline{x}_i^1,\overline{x}_
i^2)$ edge, depending on whether the variable appears 
negated in the clause or not. Thus, this flow travels through an edge representing an fulfilled literal of its clause in the assignment derived from the edge orientation. Thus, the instance has a satisfying assignment, making it a YES-instance. \qed
\end{proof}

\begin{theorem}
 \label{thm:hardness:mc:quickest}
 The quickest multicommodity contraflow problem cannot be approximated, unless $P = NP$. This holds even in the case of zero transit times.
\end{theorem}

\begin{proof}
 Consider an instance of 3-SAT, given by a set of $\ell$ clauses $C = \set{c_1,\dots,c_\ell}$ on $k$ variables $x_1, \dots, x_k$. We create a corresponding instance of the quickest multicommodity contraflow problem as follows:
 \begin{enumerate}
  \item We create a super sink $c^-$ and for each clause $c_i$ we create a node $c_i$, 
  \item for each variable $x_i$, we create nodes $x_i^1$, $x_i^2$, $\overline{x}_i^1$, $\overline{x}_i^2$, $d_i^-$, $\overline{d}_i^-$, $d_i^+$ and $\hat{d}_i^+$.
  \item For each clause $c_i = x_{i_1} \vee x_{i_2} \vee \overline{x}_{i_3}$ we create edges $\set{c_i,x_{i_1}^1}$, $\set{c_i,x_{i_2}^1}$ and $\set{c_i,\overline{x}_{i_3}^1}$,
  \item for each variable $x_i$, we create edges $\set{d_i^+,x_i^2}$, $\set{d_i^+,\overline{x}_i^2}$, $\set{x_i^1,d_i^-}$, $\set{\overline{x}_i^1,\overline{d}_i^-}$, $\set{x_i^2,c^-}$, $\set{\overline{x}_i^2,c^-}$, $\set{x_i^1,x_i^2}$, $\set{\overline{x}_i^1,\overline{x}_i^2}$, $\set{\hat{d}_i^+,d_i^-}$ and $\set{\hat{d}_i^+,\overline{d}_i^-}$.
  \item Capacities are set to $C^2$ for edges leaving $\hat{d}_i^+$, to $C$ for the other edges completely inside a variable block, to $1$ for edges entering a variable block and $\ell$ for all other edges,
  \item supplies are 1 for each clause node $c_i$, $C$ for each $d_i^+$ node, $C^2$ for each $\hat{d}_i^+$ node and zero for all other nodes,
  \item demands are $-\frac{1}{2}(C^2+C)$ for each $d_i^-$, $\overline{d}_i^-$. There is a commodity for the four $d_i$ nodes of each variable, and each clause node has supply of an own commodity, and the supersink gets a demand of $-1$ for each clause commodity.
 \end{enumerate}
 Notice that the resulting network -- an example of which is depicted in Fig.~\ref{figure3} -- has $\ell + 8k + 1$ nodes and $3\ell + 10k$ edges, which is polynomial in the size of the 3-SAT instance.
 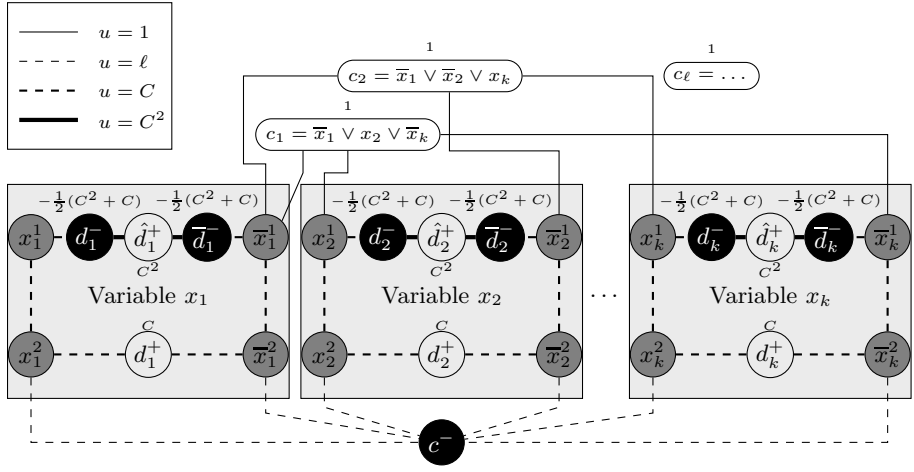
\begin{figure}[tb]  
  \centering
  \begin{tikzpicture}[scale=0.78,xscale=0.8]
   \extvariableblock{0}{0}{1}
   \extvariableblock{6.25}{0}{2}
   \node (dots) at (9.75,-1.00) {$\dots$};
   \extvariableblock{13.25}{0}{k};
   \node[draw, rounded rectangle, fill=white,label={[label distance=-0mm]90:\tiny $1$}] (cl1) at (4.25,1.75) {\scriptsize $c_1 = \overline{x}_1 \vee x_2 \vee \overline{x}_k$};   
   \draw[longEdge,thin] (cl1.200) -- ++(0,-0.375) -- (b1.45);
   \draw[longEdge,thin] (cl1.south) -- ++(0,-0.375) -| (a2.north);
   \draw[longEdge,thin] (cl1.east) -| (bk.north);
   \node[draw, rounded rectangle, fill=white,label={[label distance=-0mm]90:\tiny $1$}] (cl2) at (6.0,2.75) {\scriptsize $c_2 = \overline{x}_1 \vee \overline{x}_2 \vee x_k$};   
   \draw[longEdge,thin] (cl2.west) -- ++(-2,0) -- ++(0,-1.5) -| (b1.north);   
   \draw[longEdge,thin] (cl2.320) -- +(0,-1.00) -| (b2.north);
   \draw[longEdge,thin] (cl2.east) -| (ak.north);
   \node[draw, rounded rectangle, fill=white,label={[label distance=-0mm]90:\tiny $1$}] (clk) at (12.0,2.75) {\scriptsize $c_\ell = \dots$};
   \node[variableSource,minimum size=6mm] (s) at (6.25,-3.5) {\small $c^-$};
   \draw[shortEdge,thin] (c1) |- (s);
   \draw[shortEdge,thin] (d1) -- ++(0,-1) -- (s);
   \draw[shortEdge,thin] (c2) -- ++(0,-0.9) -- (s);
   \draw[shortEdge,thin] (d2) -- ++(0,-0.9) -- (s);
   \draw[shortEdge,thin] (ck) -- ++(0,-1) -- (s);
   \draw[shortEdge,thin] (dk) |- (s);
   \draw[draw=black, fill=white] (-3.0,4.0) rectangle (0.5,1.5);
   \draw[longEdge,thin] (-2.75,3.5) -- (-1.5,3.5); \node[anchor=west] (text) at (-1.25,3.5) {\scriptsize $u = 1$};
   \draw[shortEdge,thin] (-2.75,3.0) -- (-1.5,3.0); \node[anchor=west] (text) at (-1.25,3.0) {\scriptsize $u = \ell$};
   \draw[longEdge,thick,dashed] (-2.75,2.5) -- (-1.5,2.5); \node[anchor=west] (text) at (-1.25,2.5) {\scriptsize $u = C$};
   \draw[longEdge] (-2.75,2.0) -- (-1.5,2.0);  \node[anchor=west] (text) at (-1.25,2.0) {\scriptsize $u = C^2$};
  \end{tikzpicture}
  \caption{The quickest multicommodity contraflow flow problem instance.\label{figure3}}
 \end{figure}

 \paragraph{YES-Instance $\to$ 1 time unit required.}

 If the 3-SAT instance is a YES-instance, then there is a variable assignment $x_i \in \set{0,1}$, $i = 1,\dots,k$ fulfilling all clauses. We use this assignment to define an orientation of the edges in our network. If a variable $x_i$ is assigned a value of 0, we orient the edge $\set{x_i^1,x_i^2}$ as $(x_i^2,x_i^1)$ and $\set{\overline{x}_i^1,\overline{x}_i^2}$ as $(\overline{x}_i^1,\overline{x}_i^2)$; if $x_i$ is assigned the value $1$, we orient edge $\set{\overline{x}_i^1,\overline{x}_i^2}$ as $(\overline{x}_i^2,\overline{x}_i^1)$ and $\set{x_i^1,x_i^2}$ as $(x_i^1,x_i^2)$. All other edges are oriented away from sources and towards sinks.
 Notice that:
 \begin{enumerate}
  \item Each clause commodity of a clause $c_i = x_{i_1} \vee x_{i_2} \vee \overline{x}_{i_3}$ can send flow to the nodes $x^1_{i_1}, x^1_{i_2}, \overline{x}^1_{i_3}$.
  \item Each clause is satisfied by our assignment, so there is a literal in each clause that is true.
  \item For this literal $x_i$, there is an edge directed from $x_i^1$ to $x_i^2$ (or $\overline{x}_i^1$ to $\overline{x}_i^2$, respectively).
  \item By construction of the instance, there is a demand for this clause commodity in $c^-$, which can be reached from $x_i^2$ / $\overline{x}_i^2$.
 \end{enumerate}
 Therefore we can fulfill the demand of a clause commodity if it has satisfied literals in our assignment, which it does.  
 Now we need to consider the variable commodities. Since our assignment can only set $x_i$ to either 1 or 0, one of the edges $\set{x_i^1,x_i^2}$ and $\set{\overline{x}_i^1,\overline{x}_i^2}$ has been oriented as $(x_i^2,x_i^1)$ or $(\overline{x}_i^2,\overline{x}_i^1)$, respectively, in each variable block. This creates a path to send $C$ flow units from $d_i^+$ to one of its sinks, and the remaining demands can be covered by supply from $\hat{d}_i^+$. Since the transit times are zero, all of this can be done in 1 time unit.

 \paragraph{NO-instance $\to$ $\Theta(\frac{C}{\ell})$ time units required.}
 In a NO-instance, there is no variable assignment that satisfies all clauses. This means that we need either to orient $\set{x_i^1,x_i^2}$ as $(x_i^1,x_i^2)$ and $\set{\overline{x}_i^1,\overline{x}_i^2}$ as $(\overline{x}_i^1,\overline{x}_i^2)$, or we need to have a clause commodity use the wrong edge in a variable block (\ie, the one of the literal not contained in the clause). If we do the former, this means that we have to route the $C$ units of supply from $d_i^+$ over $c^-$, which requires at least $C/(2\ell)$ time units because of the capacities. If we do the latter, we need to switch either the direction of one of the outgoing edges of $\hat{d}_i^+$, once again causing at least $C$ time units to be necessary or we need to switch the direction of one of the incoming edges to the variable block, with a similar result. \qed 
\end{proof}

\paragraph{\textbf{Acknowledgements.}} We thank the anonymous reviewers for their helpful comments.

\bibliography{literature}

\begin{thebibliography}{10}

\bibitem{BaumannKoehler07}
N.~Baumann and E.~K{\"o}hler.
\newblock Approximating earliest arrival flows with flow-dependent transit
  times.
\newblock {\em Discrete Applied Mathematics}, 155:161--171, 2007.

\bibitem{BDK93}
R.~E. Burkard, K.~Dlaska, and B.~Klinz.
\newblock The quickest flow problem.
\newblock {\em Mathematical Methods of Operations Research}, 37:31--58, 1993.

\bibitem{FLSK07}
L.~Fleischer and M.~Skutella.
\newblock Quickest flows over time.
\newblock {\em SIAM Journal on Computing}, 36:1600--1630, 2007.

\bibitem{FLTAR98}
L.~K. Fleischer and {\'E}.~Tardos.
\newblock Efficient continuous-time dynamic network flow algorithms.
\newblock {\em Operations Research Letters}, 23:71--80, 1998.

\bibitem{FoFu:62}
L.~R. Ford and D.~R. Fulkerson.
\newblock {\em {Flows in Networks}}.
\newblock Princeton University Press, Princeton, New Jersey, 1962.

\bibitem{GAL59}
D.~Gale.
\newblock Transient flows in networks.
\newblock {\em Michigan Mathematical Journal}, 6:59--63, 1959.

\bibitem{GrKaScSc:12}
M.~Gro\ss{}, J.-P.~W. Kappmeier, D.~R. Schmidt, and M.~Schmidt.
\newblock Approximating earliest arrival flows in arbitrary networks.
\newblock In L.~Epstein and P.~Ferragina, editors, {\em Algorithms -- ESA
  2012}, volume 7501 of {\em Lecture Notes in Computer Science}, pages
  551--562. Springer Berlin Heidelberg, 2012.

\bibitem{MatEtAl:11}
M.~Hausknecht, T.-C. Au, P.~Stone, D.~Fajardo, and T.~Waller.
\newblock Dynamic lane reversal in traffic management.
\newblock In {\em 14th International IEEE Conference on Intelligent
  Transportation Systems (ITSC)}, pages 1929--1934, 2011.

\bibitem{Hirsch1989}
M.~D. Hirsch, C.~H. Papadimitriou, and S.~A. Vavasis.
\newblock Exponential lower bounds for finding brouwer fix points.
\newblock {\em Journal of Complexity}, 5:379--416, 1989.

\bibitem{HOPTAR00}
B.~Hoppe and {\'E}.~Tardos.
\newblock The quickest transshipment problem.
\newblock {\em Mathematics of Operations Research}, 25:36--62, 2000.

\bibitem{Ho:95}
B.~E. Hoppe.
\newblock {\em {Efficient Dynamic Network Flow Algorithms}}.
\newblock PhD thesis, Cornell University, 1995.

\bibitem{KiSh:05}
S.~Kim and S.~Shekhar.
\newblock {Contraflow network reconfiguration for evaluation planning: A
  summary of results}.
\newblock In {\em {Proceedings of the 13th Annual ACM International Workshop on
  Geographic Information Systems}}, pages 250--259, 2005.

\bibitem{KlinzWoeginger04}
B.~Klinz and G.~J. Woeginger.
\newblock Minimum cost dynamic flows: The series parallel case.
\newblock {\em Networks}, 43:153--162, 2004.

\bibitem{Papadimitriou1994}
C.~H. Papadimitriou.
\newblock On the complexity of the parity argument and other inefficient proofs
  of existence.
\newblock {\em Journal of Computer and System Sciences}, 48:498--532, 1994.

\bibitem{RebEtAl:10}
S.~Rebennack, A.~Arulselvan, L.~Elefteriadou, and P.~M. Pardalos.
\newblock {Complexity analysis for maximum flow problems with arc reversals}.
\newblock {\em Journal of Combinatorial Optimization}, 19:200--216, 2010.

\bibitem{Ro:39}
H.~E. Robbins.
\newblock A theorem on graphs, with an application to a problem of traffic
  control.
\newblock {\em The American Mathematical Monthly}, 46:281--283, 1939.

\bibitem{-S09}
M.~Skutella.
\newblock An introduction to network flows over time.
\newblock In W.~Cook, L.~{Lov\'{a}sz}, and J.~Vygen, editors, {\em Research
  Trends in Combinatorial Optimization}, pages 451--482. Springer, 2009.

\bibitem{Tjandra2003}
S.~A. Tjandra.
\newblock {\em Dynamic network optimization with application to the evacuation
  problem}.
\newblock PhD thesis, Technical University of Kaiserslautern, 2003.

\bibitem{TuZi:04}
H.~Tuydes and A.~Ziliaskopoulos.
\newblock Network re-design to optimize evacuation contraflow.
\newblock In {\em Proceedings of the 83rd Annual Meeting of the Transportation
  Research Board}, Washington, DC, 2004.

\bibitem{TuZi:06}
H.~Tuydes and A.~Ziliaskopoulos.
\newblock {Tabu-based heuristic approach for optimization of network evacuation
  contraflow}.
\newblock {\em Transportation Research Record}, 1964:157--168, 2006.

\bibitem{Wo:01}
B.~Wolshon.
\newblock One-way-out: Contraflow freeway operation for hurricane evacuation.
\newblock {\em Natural Hazards Review}, 2:105--112, 2001.

\bibitem{WoUrLe:02}
B.~Wolshon, E.~Urbina, and M.~Levitan.
\newblock {National review of hurricane evacuation plans and policies}.
\newblock Technical report, {LSU Hurricane Center, Louisiana State University,
  Baton Rouge, Louisiana}, 2002.

\end{thebibliography}

\end{document}